\documentclass[10pt,reqno]{amsart}
\usepackage[hypertex]{hyperref}

\usepackage{amsmath,amsthm}
\usepackage{amssymb}
\usepackage{euscript}
\usepackage[frame,ps,matrix,arrow,curve,rotate,all,2cell,tips]{xy}

\numberwithin{equation}{subsection}

\newcommand{\sqsp}{\renewcommand{\baselinestretch}{1.1}\tiny\normalsize}

\newcommand{\relphantom[1]}{\mathrel{\phantom{#1}}}

\newtheorem{theorem}[subsection]{Theorem}
\newtheorem{lemma}[subsection]{Lemma}
\newtheorem{proposition}[subsection]{Proposition}
\newtheorem{corollary}[subsection]{Corollary}

\theoremstyle{definition}
\newtheorem{definition}[subsection]{Definition}

\newtheorem{remark}[subsection]{Remark}

\newcommand{\fg}{\mathfrak{g}}
\newcommand{\fh}{\mathfrak{h}}
\newcommand{\bC}{\mathbf{C}}
\newcommand{\bk}{\mathbf{k}}

\newcommand{\Qbar}{\overline{\mathbb{Q}}}
\newcommand{\bQ}{\mathbb{Q}}

\def\cyclicsum{{\circlearrowleft}}

\DeclareMathOperator{\Hom}{Hom}
\DeclareMathOperator{\Aut}{Aut}
\DeclareMathOperator{\ad}{ad}
\DeclareMathOperator{\Gal}{Gal}

\def\alphaA{{\begin{pmatrix}1 & c & a\\
-2ab & b & -a^2b\\
-2b^{-1}c & -b^{-1}c^2 & b^{-1}\end{pmatrix}}}

\def\alphaB{{\begin{pmatrix}-1 & c & a\\
2b^{-1}c & -b^{-1}c^2 & b^{-1}\\
2ab & b & -a^2b\end{pmatrix}}}

\def\alphaC{{\begin{pmatrix}c & a & \dfrac{1-c^2}{4a}\\
b & \dfrac{ab}{c-1} & \dfrac{b(1-c)}{4a}\\
\dfrac{1-c^2}{b} & \dfrac{a(1-c)}{b} & \dfrac{(c^2-1)(c+1)}{4ab}\end{pmatrix}}}

\begin{document}

\title{The classical Hom-Yang-Baxter equation and Hom-Lie bialgebras}
\author{Donald Yau}

\begin{abstract}
Motivated by recent work on Hom-Lie algebras and the Hom-Yang-Baxter equation, we introduce a twisted generalization of the classical Yang-Baxter equation (CYBE), called the classical Hom-Yang-Baxter equation (CHYBE). We show how an arbitrary solution of the CYBE induces multiple infinite families of solutions of the CHYBE.  We also introduce the closely related structure of Hom-Lie bialgebras, which generalize Drinfel'd's Lie bialgebras.  In particular, we study the questions of duality and cobracket perturbation and the sub-classes of coboundary and quasi-triangular Hom-Lie bialgebras.
\end{abstract}



\address{Department of Mathematics\\
    The Ohio State University at Newark\\
    1179 University Drive\\
    Newark, OH 43055, USA}
\email{dyau@math.ohio-state.edu}

\date{\today}
\maketitle

\sqsp

\section{Introduction}


The classical Yang-Baxter equation (CYBE), also known as the classical triangle equation, was introduced by Sklyanin \cite{skl1,skl2} in the context of quantum inverse scattering method \cite{fad1,fad2}.  The CYBE in a Lie algebra $L$ states $[r_{12},r_{13}] + [r_{12},r_{23}] + [r_{13},r_{23}] = 0$ for an $r \in L^{\otimes 2}$.  Here for $r = \sum r_1 \otimes r_2$ and $s = \sum s_1 \otimes s_2 \in L^{\otimes 2}$, the three brackets are defined as $[r_{12},s_{13}] = \sum [r_1,s_1] \otimes r_2 \otimes s_2$, $[r_{12},s_{23}] = \sum r_1 \otimes [r_2,s_1] \otimes s_2$, and $[r_{13},s_{23}] = \sum r_1 \otimes s_1 \otimes [r_2,s_2]$.  Such an $r$ is said to be a solution of the CYBE or a \emph{classical $r$-matrix}.  The CYBE is one of several equations collectively known as the Yang-Baxter equations (YBE), which were first introduced by Baxter, McGuire, and Yang \cite{baxter,baxter2,yang} in statistical mechanics.  The various forms of the YBE and some of their uses in physics are summarized in \cite{perk}.

The CYBE is closely related to many topics in mathematical physics, including Hamiltonian structures \cite{cher1,cher2,gc,gd}, Kac-Moody algebras \cite{dri87,kac}, Poisson-Lie groups, Poisson-Hopf algebras, quantum groups, Hopf algebras, and Lie bialgebras \cite{dri83,dri85,dri87,dri89,sts1,sts2}. There are many known solutions of the CYBE.  For example, each complex semi-simple Lie algebra has a non-trivial classical $r$-matrix \cite{bd1,bd2,bd3,dri83,dri87,jimbo1,jimbo2}.  There are numerous articles in the literature that deal with classical $r$-matrices, e.g., \cite{bm,feld,feld2,michaelis2,nt,stolin1,stolin2}, to name a few.  Classification of solutions of the CYBE, possibly in parametrized form, can be found in \cite{bd1,bd2,bd3,ks}.


The purpose of this paper is to study a twisted generalization of the CYBE and the closely related object of Hom-Lie bialgebra, both of which are motivated by recent work on Hom-Lie algebras and the Hom-Yang-Baxter equation (HYBE).  Before we define our twisted CYBE, let us recall some basic elements about Hom-Lie algebras and the HYBE.  A \emph{Hom-Lie algebra} $(L,[-,-],\alpha)$ consists of a vector space $L$, an anti-symmetric bilinear operation $[-,-] \colon L^{\otimes 2} \to L$, and a linear map $\alpha \colon L \to L$ such that the following conditions are satisfied: (1) $\alpha \circ [-,-] = [-,-] \circ \alpha^{\otimes 2}$ (multiplicativity) and (2) the \emph{Hom-Jacobi identity},
\begin{equation}
\label{eq:HomJacobi}
[[x,y],\alpha(z)] + [[z,x],\alpha(y)] + [[y,z],\alpha(x)] = 0
\end{equation}
for $x,y,z \in L$.  If $\alpha = Id$, then the Hom-Jacobi identity reduces to the usual Jacobi identity, and we have a Lie algebra.  Hom-Lie algebras were introduced in \cite{hls} (without multiplicativity) to describe the structures on some $q$-deformations of the Witt and the Virasoro algebras.  Earlier precursors of Hom-Lie algebras can be found in \cite{hu,liu}.  Hom-Lie algebras are closely related to discrete and deformed vector fields \cite{hls,ls,ls2}.  Further studies of Hom-Lie algebras and related Hom type algebras can be found in \cite{atm,fg,fg2,ms,ms2,ms3,ms4,yau,yau2,yau3,yau4,yau5,yau6}.


In \cite{yau5} the author introduced the Hom-Yang-Baxter equation (HYBE) as a Hom type generalization of the YBE. The HYBE states
\begin{equation}
\label{eq:HYBE}
(\alpha \otimes B) \circ (B \otimes \alpha) \circ (\alpha \otimes B) = (B \otimes \alpha) \circ (\alpha \otimes B) \circ (B \otimes \alpha),
\end{equation}
where $V$ is a vector space, $\alpha \colon V \to V$ is a linear map, and $B \colon V^{\otimes 2} \to V^{\otimes 2}$ is a bilinear (not-necessarily invertible) map that commutes with $\alpha^{\otimes 2}$.  The HYBE reduces to the usual YBE when $\alpha = Id$.  Several classes of solutions of the HYBE were constructed in \cite{yau5}, including those associated to Hom-Lie algebras and Drinfel'd's (dual) quasi-triangular bialgebras \cite{dri87}.  It is also shown in \cite{yau5} that, just like solutions of the YBE, each solution of the HYBE gives rise to operators that satisfy the braid relations \cite{artin,artin2}.  With an additional invertibility condition, these operators give a representation of the braid group.  Additional solutions of the HYBE were constructed in \cite{yau6}, some of which are closely related to the quantum enveloping algebra of $sl(2)$ \cite{dri85,dri87,jimbo,kassel,kr,majid}, the Jones-Conway polynomial \cite{conway,homfly,jones1,jones2}, and Yetter-Drinfel'd modules \cite{rt,yetter2}.


As illustrated by the definitions of Hom-Lie algebras and the HYBE, Hom type structures arise when the identity map is strategically replaced by some twisting map $\alpha$.  With this in mind, we define the \emph{classical Hom-Yang-Baxter equation} (CHYBE) in a Hom-Lie algebra $(L,[-,-],\alpha)$ as
\begin{equation}
\label{eq:chybe}
[[r,r]]^\alpha \buildrel \text{def} \over= [r_{12},r_{13}] + [r_{12},r_{23}] + [r_{13},r_{23}] = 0
\end{equation}
for $r \in L^{\otimes 2}$.  The three brackets in \eqref{eq:chybe} are defined as
\begin{equation}
\label{eq:r123}
\begin{split}
[r_{12},s_{13}] &= \sum [r_1,s_1] \otimes \alpha(r_2) \otimes \alpha(s_2),\quad
[r_{12},s_{23}] = \sum \alpha(r_1) \otimes [r_2,s_1] \otimes \alpha(s_2),\\
[r_{13},s_{23}] &= \sum \alpha(r_1) \otimes \alpha(s_1) \otimes [r_2,s_2],
\end{split}
\end{equation}
where $r = \sum r_1 \otimes r_2$ and $s = \sum s_1 \otimes s_2 \in L^{\otimes 2}$.  If $\alpha = Id$ (i.e., $L$ is a Lie algebra), then the CHYBE reduces to the CYBE: $[[r,r]]^{Id} = 0$.  In this case, a solution of the CHYBE is just a classical $r$-matrix.

We now describe our results about the CHYBE and related algebraic structures.  These results will be proved in later sections.

First we address the question of constructing solutions of the CHYBE.  We go back to Hom-Lie algebras and the HYBE for inspirations.  There is a general strategy, first used in \cite{yau2} and later in \cite{atm,fg2,ms4,yau3,yau4}, to twist an algebraic structure into a corresponding Hom type object via an endomorphism.  In particular, it is not hard to check directly that if $L$ is a Lie algebra and $\alpha \colon L \to L$ is a Lie algebra morphism, then $L_\alpha = (L,[-,-]_\alpha,\alpha)$ is a Hom-Lie algebra, where the twisted bracket $[-,-]_\alpha$ is defined as $\alpha \circ [-,-]$ \cite{yau2}.  There is a similar result about twisting a solution of the YBE into a solution of the HYBE \cite{yau6}.

In section ~\ref{sec:rmatrix}, we show that if $r \in L^{\otimes 2}$ is a solution of the CYBE in the Lie algebra $L$ and if $\alpha \colon L \to L$ is a Lie algebra morphism, then $(\alpha^{\otimes 2})^n(r)$ is a solution of the CHYBE in the Hom-Lie algebra $L_\alpha = (L,[-,-]_\alpha,\alpha)$ for each integer $n \geq 0$ (Theorem ~\ref{thm:chybe}).  In other words, each Lie algebra endomorphism and each classical $r$-matrix induces a (usually infinite) family of solutions of the CHYBE.  This gives an efficient method for constructing many solutions of the CHYBE.  We illustrate this result with the Lie algebra $sl(2)$, equipped with its standard classical $r$-matrix \eqref{eq:slr} \cite{bd1,bd2,bd3,dri83,dri87,jimbo1,jimbo2}.  We compute the solutions $(\alpha^{\otimes 2})^n(r)$ of the CHYBE for \emph{all} the Lie algebra endomorphisms $\alpha$ on $sl(2)$ (Propositions ~\ref{prop:sl1} - ~\ref{prop:sl3}), making use of the classification of these maps obtained in \cite{yau6}.  There is a distinct property of these solutions of the CHYBE that is worth mentioning.  In fact, on the one hand, the standard classical $r$-matrix on $sl(2)$ lies in a two-dimensional subspace of $sl(2)^{\otimes 2}$.  On the other hand, all nine dimensions in $sl(2)^{\otimes 2}$ are involved in describing $(\alpha^{\otimes 2})^n(r)$ for the various endomorphisms on $sl(2)$ (Remark ~\ref{remark:sl}).


As Drinfel'd explains in \cite{dri87}, classical $r$-matrices often arise in conjunction with the richer structure of a Lie bialgebra, which consists of a Lie algebra that also has a Lie coalgebra structure, in which the cobracket is a $1$-cocycle in Chevalley-Eilenberg cohomology (\eqref{eq:compatibility} with $\alpha = Id$).  Dualizing the definition of a Hom-Lie algebra, one can define a Hom-Lie coalgebra (Definition ~\ref{def:hlc}), as was done in  \cite{ms4}.  We generalize Drinfel'd's Lie bialgebra and define \emph{Hom-Lie bialgebra} (Definition ~\ref{def:hlb}), in which the cobracket satisfies an analogous cocycle condition \eqref{eq:compatibility}.  In fact, the condition \eqref{eq:compatibility} says exactly that the cobracket in a Hom-Lie bialgebra is a $1$-cocycle in Hom-Lie algebra cohomology (Remark ~\ref{rk:cocycle}).

In section ~\ref{sec:hl}, we show that an arbitrary Lie bialgebra can be twisted into a Hom-Lie bialgebra via any Lie bialgebra endomorphism (Theorem ~\ref{thm:hlbdeform}).  This gives an efficient method for constructing Hom-Lie bialgebras.  When the twisting maps are invertible, we give a group-theoretic criterion under which two such Hom-Lie bialgebras are isomorphic (Proposition ~\ref{prop:twist} - Corollary ~\ref{cor2:twist}).  Using these results, we observe that \emph{uncountably many} mutually non-isomorphic Hom-Lie bialgebras $sl(2)_\alpha$ can be constructed this way from the Lie bialgebra $sl(2)$ (Corollary ~\ref{cor1:HLsl} and Corollary ~\ref{cor2:HLsl}).  We also show that Hom-Lie bialgebras have a self-dual property.  Namely, for a finite dimensional Hom-Lie bialgebra, its linear dual is also a Hom-Lie bialgebra with the dual bracket, cobracket, and $\alpha$ (Theorem ~\ref{thm;duality}).  These results are illustrated with the example of $sl(2)$, regarded as a Lie bialgebra (Proposition ~\ref{prop:slbimorphism} - Corollary ~\ref{cor:dual}).


The connection between classical $r$-matrices and Lie bialgebras comes from the sub-classes of coboundary and quasi-triangular Lie bialgebras \cite{dri87}.  Briefly, a \emph{coboundary Lie bialgebra} $(L,r)$ is a Lie bialgebra $L$ in which the cobracket $\Delta$ is a $1$-coboundary in Chevalley-Eilenberg cohomology, i.e., $\Delta = \ad(r)$ for some $r \in L^{\otimes 2}$ (see \eqref{eq:Delta=ad} with $\alpha = Id$).  A \emph{quasi-triangular Lie bialgebra} is a coboundary Lie bialgebra $(L,r)$ in which $r$ is also a classical $r$-matrix (i.e., $r$ satisfies \eqref{eq:chybe} with $\alpha = Id$).

In section ~\ref{sec;cob}, replacing the identity map by a general twisting map $\alpha$, we define \emph{coboundary Hom-Lie bialgebras} and \emph{quasi-triangular Hom-Lie bialgebras} (Definition ~\ref{def:coboundary}) as analogous sub-classes of Hom-Lie bialgebras.  In a coboundary Hom-Lie bialgebra, the cobracket $\Delta$ is a $1$-coboundary $\ad(r)$ in Hom-Lie algebra cohomology (Remark ~\ref{rk:delta}).  A quasi-triangular Hom-Lie bialgebra is a coboundary Hom-Lie bialgebra in which $r$ is also a solution of the CHYBE \eqref{eq:chybe}. These Hom type objects can be constructed by twisting coboundary and quasi-triangular Lie bialgebras, respectively, via suitable endomorphisms (Theorem ~\ref{thm:cob}).  For example, the Hom-Lie bialgebras $sl(2)_\alpha$ in Corollary ~\ref{cor1:HLsl} are all quasi-triangular Hom-Lie bialgebras (Corollary ~\ref{cor:sl2quasi}).  We then describe conditions under which a Hom-Lie algebra $L$ and an element $r \in L^{\otimes 2}$ give a coboundary or a quasi-triangular Hom-Lie bialgebra (Theorem ~\ref{thm:cob-char} and Corollary ~\ref{cor:cob}).  Going a step further, given a coboundary Hom-Lie bialgebra, we give several characterizations of when it is a quasi-triangular Hom-Lie bialgebra (Theorem ~\ref{thm:quasi}), i.e., when $r$ is a solution of the CHYBE.

In section ~\ref{sec:perturb}, we study cobracket perturbation in (quasi-triangular) Hom-Lie bialgebras, following the perturbation theory initiated by Drinfel'd for quasi-Hopf algebras \cite{dri83b,dri89b,dri90,dri91,dri92}.  In particular, we describe conditions under which the cobracket in a Hom-Lie bialgebra can be perturbed by a coboundary to give another Hom-Lie bialgebra (Theorem ~\ref{thm:perturb} and Corollary ~\ref{cor1:perturb}).  There is a similar result about cobracket perturbation in a quasi-triangular Hom-Lie bialgebra (Corollary ~\ref{cor:perturb}).

\section{Solutions of the CHYBE from classical $r$-matrices}
\label{sec:rmatrix}


\subsection{Conventions}
Throughout the rest of this paper, we work over a fixed field $\bk$ of characteristic $0$.  Vector spaces, tensor products, linearity, and $\Hom$ are all meant over $\bk$. If $f \colon V \to V$ is a linear self-map on a vector space $V$, then $f^n \colon V \to V$ denotes the composition $f \circ \cdots \circ f$ of $n$ copies of $f$, with $f^0 = Id$. For an element $r = \sum r_1 \otimes r_2 \in V^{\otimes 2}$, the summation sign will often be omitted in computations to simplify the typography.


The first result shows that, given a Lie algebra endomorphism, each classical $r$-matrix induces a (usually infinite) family of solutions of the CHYBE.  Afterwards, we will illustrate this result with the Lie algebra $sl(2)$.

\begin{theorem}
\label{thm:chybe}
Let $L$ be a Lie algebra, $r \in L^{\otimes 2}$ be a solution of the CYBE, and $\alpha \colon L \to L$ be a Lie algebra morphism.  Then for each integer $n \geq 0$, $(\alpha^{\otimes 2})^n(r)$ is a solution of the CHYBE \eqref{eq:chybe} in the Hom-Lie algebra $L_\alpha = (L,[-,-]_\alpha = \alpha \circ [-,-],\alpha)$.
\end{theorem}

\begin{proof}
We already mentioned in the introduction that $L_\alpha$ is a Hom-Lie algebra, a fact that is not hard to check directly \cite{yau2}.  (In fact, the Hom-Jacobi identity for $[-,-]_\alpha$ is $\alpha^2$ applied to the Jacobi identity of $[-,-]$.)  It remains to show that $(\alpha^{\otimes 2})^n(r)$ satisfies the CHYBE \eqref{eq:chybe} in the Hom-Lie algebra $L_\alpha$, i.e., $[[(\alpha^{\otimes 2})^n(r),(\alpha^{\otimes 2})^n(r)]]^\alpha = 0$.

Write $r = \sum r_1 \otimes r_2 \in L^{\otimes 2}$, and let $r' = \sum r_1' \otimes r_2'$ be another copy of $r$.  Using $\alpha([-,-]) = [-,-] \circ \alpha^{\otimes 2}$ and the definition $[-,-]_\alpha = \alpha([-,-])$, we have
\[
\begin{split}
[[(\alpha^{\otimes 2})^n(r),(\alpha^{\otimes 2})^n(r)]]^\alpha
&= [\alpha^n(r_1),\alpha^n(r_1')]_\alpha \otimes \alpha(\alpha^n(r_2)) \otimes \alpha(\alpha^n(r_2'))\\
&\relphantom{} + \alpha(\alpha^n(r_1)) \otimes [\alpha^n(r_2),\alpha^n(r_1')]_\alpha \otimes \alpha(\alpha^n(r_2'))\\
&\relphantom{} + \alpha(\alpha^n(r_1)) \otimes \alpha(\alpha^n(r_1')) \otimes [\alpha^n(r_2),\alpha^n(r_2')]_\alpha\\
&= \alpha^{n+1}([r_1,r_1']\otimes r_2 \otimes r_2' + r_1 \otimes [r_2,r_1'] \otimes r_2' + r_1 \otimes r_1' \otimes [r_2,r_2'])\\
&= \alpha^{n+1}([[r,r]]^{Id}) = 0.
\end{split}
\]
In the last line above, the CYBE $[[r,r]]^{Id}=0$ is taking place in the Lie algebra $L$.
\end{proof}


Theorem ~\ref{thm:chybe} is useful as long as we can compute endomorphisms of interesting Lie algebras and the induced solutions $(\alpha^{\otimes 2})^n(r)$ of the CHYBE.  Let us illustrate Theorem ~\ref{thm:chybe} with the complex Lie algebra $sl(2)$ \cite[p.13-14]{jac}.  This is the three-dimensional complex Lie algebra with basis $\{H,X_+,X_-\}$, whose Lie bracket is determined by
\begin{equation}
\label{slbracket}
[X_+,X_-] = H \quad\text{and}\quad
[H,X_\pm] = \pm 2X_\pm.
\end{equation}
Our notations follow \cite[Example 8.1.10]{majid}.  The Lie algebra $sl(2)$ has a standard non-trivial solution of the CYBE defined as \cite{bd1,bd2,bd3,dri83,dri87,jimbo1,jimbo2}
\begin{equation}
\label{eq:slr}
r = X_+ \otimes X_- + \frac{1}{4}H \otimes H.
\end{equation}
Note that this $r$ lies in a two-dimensional subspace of $sl(2)^{\otimes 2}$.  We will describe all the solutions of the CHYBE in the Hom-Lie algebra $sl(2)_\alpha$ of the form $(\alpha^{\otimes 2})^n(r)$.

Let us first recall from \cite{yau6} the classification of Lie algebra endomorphisms on $sl(2)$.  With respect to the basis $\{H,X_+,X_-\}$ of $sl(2)$, a non-zero linear map $\alpha \colon sl(2) \to sl(2)$ is a Lie algebra morphism if and only if its matrix has one of the following three forms (where $a,b,c$ are complex numbers):
\begin{subequations}
\allowdisplaybreaks
\begin{align}
\alpha_1 & = \alphaA \quad\text{with $b \not= 0$ and $ac = 0$},\label{alpha1}\\
\alpha_2 & = \alphaB \quad\text{with $b \not= 0$ and $ac = 0$},\label{alpha2}\\
\alpha_3 & = \alphaC \quad\text{with $ab \not= 0$ and $c \not= \pm 1$}.\label{alpha3}
\end{align}
\end{subequations}
Now we describe the solutions $(\alpha^{\otimes 2})^n(r)$ of the CHYBE in the Hom-Lie algebra $sl(2)_\alpha$ for all these maps $\alpha$.  For the rest of this section, $r \in sl(2)^{\otimes 2}$ is the classical $r$-matrix in \eqref{eq:slr}.

\begin{proposition}
\label{prop:sl1}
Suppose $\alpha = \alpha_1$ \eqref{alpha1}.  For $n \geq 1$, we have
\begin{equation}
\label{eq:alpha1r}
\begin{split}
(\alpha^{\otimes 2})^n(r)
&= r
+ (ab)\left\{\biggl(\sum_{i=0}^{n-1}b^i\biggr)\left(\frac{1}{2}X_+ \otimes H - 2H \otimes X_+\right) + 3ab^2d_n X_+ \otimes X_+\right\}\\
&+ cb^{-1}\left\{\biggl(\sum_{i=0}^{n-1}b^{-i}\biggr)\left(\frac{1}{2}H \otimes X_- - 2X_- \otimes H\right) + 3b^{-2}ce_n X_- \otimes X_-\right\},
\end{split}
\end{equation}
where $d_1 = e_1 = 0$,
\begin{equation}
\label{eq:den}
d_{n+1} = b^2d_n + \sum_{i=0}^{n-1}b^i, \quad\text{and}\quad
e_{n+1} = b^{-2}e_n + \sum_{i=0}^{n-1}b^{-i}.
\end{equation}
\end{proposition}

For the maps $\alpha_2$ and $\alpha_3$, let us introduce a notation that will simplify the typography below.  If $Y,Z \in \{H,X_+,X_-\}$ are basis elements of $sl(2)$, we set
\begin{equation}
\label{absolute}
|Y \otimes Z| = Y \otimes Z - Z \otimes Y \in sl(2)^{\otimes 2}.
\end{equation}
For example, $|X_+ \otimes X_-| = X_+ \otimes X_- - X_- \otimes X_+$.

\begin{proposition}
\label{prop:sl2}
Suppose $\alpha = \alpha_2$ \eqref{alpha2}.  Using the notation ~\eqref{absolute}, for $n \geq 0$, we have
\begin{equation}
\label{eq:alpha2r}
(\alpha^{\otimes 2})^n(r)
= r + j_n|X_+\otimes X_-| + k_n|H\otimes X_+| + l_n|H\otimes X_-|,
\end{equation}
where $j_0 = k_0 = l_0 = 0$,
\[
\begin{split}
j_{n+1} &=
-1-j_n+2(a+c)k_n,\\
k_{n+1} &=
\begin{cases}b^{-1}\left(\dfrac{c}{2} + cj_n - c^2k_n - l_n\right) &\text{ if $a=0$},\\
b^{-1}l_n &\text{ if $c=0$},\end{cases}\\
l_{n+1} &=
-b\left(\dfrac{a}{2} + aj_n + k_n - a^2l_n\right).
\end{split}
\]
\end{proposition}

\begin{proposition}
\label{prop:sl3}
Suppose $\alpha = \alpha_3$ \eqref{alpha3}.  Using the notation ~\eqref{absolute}, for $n \geq 0$, we have
\begin{equation}
\label{eq:alpha3r}
(\alpha^{\otimes 2})^n(r)
= r + p_n|X_+\otimes X_-| + q_n|H\otimes X_+| + s_n|H\otimes X_-|,
\end{equation}
where $p_0 = q_0 = s_0 = 0$,
\[
\begin{split}
p_{n+1} &=
\frac{c-1}{2} + cp_n + 2aq_n + \left(\frac{c^2-1}{2a}\right)s_n,\\
q_{n+1} &=
\frac{b}{4}\left(1 + 2p_n + \left(\frac{4a}{c-1}\right)q_n + \left(\frac{c-1}{a}\right)s_n\right),\\
s_{n+1} &=
\frac{c^2-1}{4b}\left(1 + 2p_n + \left(\frac{4a}{c+1}\right)q_n + \left(\frac{c+1}{a}\right)s_n\right).
\end{split}
\]
\end{proposition}

\begin{remark}
\label{remark:sl}
Note that in Proposition ~\ref{prop:sl1}, $(\alpha^{\otimes 2})^n(r)$ in general lies in a five-dimensional subspace of $sl(2)^{\otimes 2}$, since either $a = 0$ or $c = 0$.  Moreover, eight of the nine basis elements in $sl(2)^{\otimes 2}$ are used in \eqref{eq:alpha1r}.  Likewise, in Propositions ~\ref{prop:sl2} and ~\ref{prop:sl3}, $(\alpha^{\otimes 2})^n(r)$ in general lies in a seven-dimensional subspace of $sl(2)^{\otimes 2}$.
\end{remark}

Propositions ~\ref{prop:sl1} - ~\ref{prop:sl3} are proved by very similar induction arguments, so we will only give the proof of Proposition ~\ref{prop:sl1} in details.

\begin{proof}[Proof of Proposition ~\ref{prop:sl1}]
In $\alpha = \alpha_1$ \eqref{alpha1}, either $a=0$ or $c=0$.  Suppose $c=0$.  In this case, we have $\alpha(H) = H - 2abX_+$, $\alpha(X_+) = bX_+$, and $\alpha(X_-) = aH - a^2bX_+ + b^{-1}X_-$.  By direct computation, we obtain
\begin{equation}
\label{eq:slalpha1}
\begin{split}
\alpha^{\otimes 2}(r) &=
r + (ab)\left(\frac{1}{2}X_+ \otimes H - 2H\otimes X_+\right),\\
\alpha^{\otimes 2}\left(\frac{1}{2}X_+ \otimes H - 2H\otimes X_+\right) &=
b\left(\frac{1}{2}X_+ \otimes H - 2H \otimes X_+ + 3abX_+ \otimes X_+\right),\\
\alpha^{\otimes 2}(X_+ \otimes X_+) &= b^2X_+ \otimes X_+.
\end{split}
\end{equation}
In particular, \eqref{eq:alpha1r} holds when $n = 1$ and $c = 0$.  Inductively, suppose \eqref{eq:alpha1r} holds for some $n \geq 1$ (still with $c = 0$).  Using \eqref{eq:slalpha1} we have
\[
\begin{split}
(\alpha^{\otimes 2})^{n+1}(r) &=
\alpha^{\otimes 2}(r) +
(ab)\biggl(\sum_{i=0}^{n-1}b^i\biggr)\alpha^{\otimes 2}\left(\frac{1}{2}X_+ \otimes H - 2H\otimes X_+\right)
+ (ab)(3ab^2d_n)\alpha^{\otimes 2}(X_+ \otimes X_+)\\
&= r + (ab)\left(\frac{1}{2}X_+ \otimes H - 2H\otimes X_+\right)  + (ab)(3ab^2d_n)b^2X_+ \otimes X_+\\
&\relphantom{} + (ab)\biggl(\sum_{i=0}^{n-1}b^i\biggr)b\left(\frac{1}{2}X_+ \otimes H - 2H \otimes X_+ + 3abX_+ \otimes X_+\right)\\
&= r + (ab)\left(1 + \sum_{i=0}^{n-1}b^{i+1}\right)\left(\frac{1}{2}X_+ \otimes H - 2H\otimes X_+\right)\\
&\relphantom{} + (ab)(3ab^2)\left(b^2d_n + \sum_{i=0}^{n-1}b^i\right)X_+ \otimes X_+.
\end{split}
\]
Comparing this with the definition \eqref{eq:den} of $d_{n+1}$, we conclude that the formula \eqref{eq:alpha1r} holds for the case $n+1$ as well.  This proves \eqref{eq:alpha1r} when $c = 0$.  The case $a = 0$ is proved similarly.
\end{proof}

\section{Hom-Lie bialgebras}
\label{sec:hl}

In this section, we introduce Hom-Lie bialgebra, which is the Hom version of Drinfel'd's Lie bialgebra \cite{dri83,dri87}.  The connections between Hom-Lie bialgebras and the CHYBE \eqref{eq:chybe} will be explored in the next two sections.  Here we observe that Lie bialgebras can be twisted along any endomorphism to produce Hom-Lie bialgebras (Theorem ~\ref{thm:hlbdeform}).  When the twisting maps are invertible, we give a group-theoretic characterization of when two such Hom-Lie bialgebras are isomorphic (Proposition ~\ref{prop:twist} - Corollary ~\ref{cor2:twist}).  Then we show that the dual of a finite dimensional Hom-Lie bialgebra is also a Hom-Lie bialgebra (Theorem ~\ref{thm;duality}), generalizing the self-dual property of Lie bialgebras.  At the end of this section, we illustrate these results with the Lie bialgebra $sl(2)$ (Proposition ~\ref{prop:slbimorphism} - Corollary ~\ref{cor:dual}).


\subsection{Notations}
\label{notations}
Let $V$ and $W$ be vector spaces.
\begin{enumerate}
\item
Denote by $\tau \colon V \otimes W \to W \otimes V$ the twist isomorphism, $\tau(x \otimes y) = y \otimes x$.
\item
The symbol $\cyclicsum$ denotes the cyclic sum in three variables.  In other words, if $\sigma$ is the cyclic permutation $(1 \, 2 \, 3)$, then $\cyclicsum$ is the sum over $Id$, $\sigma$, and $\sigma^2$.  With this notation, the Hom-Jacobi identity ~\eqref{eq:HomJacobi} can be restated as: $\cyclicsum [-,-]\circ ([-,-]\otimes \alpha) = 0$.
\item
For a linear map $\Delta \colon V \to V^{\otimes 2}$, we use Sweedler's notation $\Delta(x) = \sum_{(x)} x_1 \otimes x_2$ for $x \in V$.  We will often omit the summation sign $\sum_{(x)}$ to simplify the typography.
\item
Denote by $V^* = \Hom(V,\bk)$ the linear dual of $V$.  For $\phi \in V^*$ and $x \in V$, we often use the adjoint notation $\langle \phi,x\rangle$ for $\phi(x) \in \bk$.
\item
For an element $x$ in a Hom-Lie algebra $(L,[-,-],\alpha)$ and $n \geq 2$, define the adjoint map $\ad_x \colon L^{\otimes n} \to L^{\otimes n}$ by
\begin{equation}
\label{eq:ad}
\ad_x(y_1 \otimes \cdots \otimes y_n) = \sum_{i=1}^n \alpha(y_1) \otimes \cdots \otimes \alpha(y_{i-1}) \otimes [x,y_i] \otimes \alpha(y_{i+1}) \cdots \otimes \alpha(y_n).
\end{equation}
Conversely, given $\gamma = y_1 \otimes \cdots \otimes y_n$, we define the map $\ad(\gamma) \colon L \to L^{\otimes n}$ by $\ad(\gamma)(x) = \ad_x(\gamma)$.
\end{enumerate}

First we define the dual objects of Hom-Lie algebras.

\begin{definition}[\cite{ms4}]
\label{def:hlc}
A \textbf{Hom-Lie coalgebra} $(L,\Delta,\alpha)$ consists of a vector space $L$, a linear self $\alpha \colon L \to L$, and a linear map $\Delta \colon L \to L^{\otimes 2}$ such that (1) $\Delta \circ \alpha = \alpha^{\otimes 2} \circ \Delta$ (co-multiplicativity), (2) $\tau \circ \Delta = -\Delta$ (anti-symmetry), and (3) $\cyclicsum (\alpha \otimes \Delta) \circ \Delta = 0$ (Hom-co-Jacobi identity).  We call $\Delta$ the \textbf{cobracket}.
\end{definition}

A Hom-Lie coalgebra with $\alpha = Id$ is exactly a Lie coalgebra \cite{michaelis}.
Just like (co)associative (co)algebras, if $(L,\Delta,\alpha)$ is a Hom-Lie coalgebra, then $(L^*,[-,-],\alpha)$ is a Hom-Lie algebra, as defined in the introduction (the paragraph containing \eqref{eq:HomJacobi}).  Here $[-,-]$ and $\alpha$ in $L^*$ are dual to $\Delta$ and $\alpha$, respectively, in $L$.  Conversely, if $(L,[-,-],\alpha)$ is a \emph{finite dimensional} Hom-Lie algebra, then $(L^*,\Delta,\alpha)$ is a Hom-Lie coalgebra, where $\Delta$ and $\alpha$ in $L^*$ are dual to $[-,-]$ and $\alpha$, respectively, in $L$.  These facts are also special cases of \cite[Propositions 4.10 and 4.11]{ms4}.


\begin{definition}
\label{def:hlb}
A \textbf{Hom-Lie bialgebra} $(L,[-,-],\Delta,\alpha)$ is a quadruple in which $(L,[-,-],\alpha)$ is a Hom-Lie algebra and $(L,\Delta,\alpha)$ is a Hom-Lie coalgebra such that the following compatibility condition holds for all $x, y \in L$:
\begin{equation}
\label{eq:compatibility}
\Delta([x,y]) = \ad_{\alpha(x)}(\Delta(y)) - \ad_{\alpha(y)}(\Delta(x)).
\end{equation}
A \textbf{morphism} $f \colon L \to L'$ of Hom-Lie bialgebras is a linear map such that $\alpha \circ f = f \circ \alpha$, $f([-,-]) = [-,-] \circ f^{\otimes 2}$, and $\Delta \circ f = f^{\otimes 2} \circ \Delta$.    An \textbf{isomorphism} of Hom-Lie bialgebras is an invertible morphism of Hom-Lie bialgebras.  Two Hom-Lie bialgebras are said to be \textbf{isomorphic} if there exists an isomorphism between them.
\end{definition}

A Hom-Lie bialgebra with $\alpha = Id$ is exactly a Lie bialgebra, as defined by Drinfel'd \cite{dri83,dri87}.  One can also use this as the definition of a Lie bialgebra.  We can unwrap the compatibility condition \eqref{eq:compatibility} as
\begin{equation}
\label{eq:compat}
\Delta([x,y]) = [\alpha(x),y_1] \otimes \alpha(y_2) + \alpha(y_1) \otimes [\alpha(x),y_2]
- \left([\alpha(y),x_1] \otimes \alpha(x_2) + \alpha(x_1) \otimes [\alpha(y),x_2]\right).
\end{equation}

\begin{remark}
\label{rk:cocycle}
The compatibility condition \eqref{eq:compatibility} is, in fact, a cocycle condition in Hom-Lie algebra cohomology \cite[section 5]{ms3}, just as it is the case in a Lie bialgebra (with Lie algebra cohomology) \cite{dri87}.  Indeed, we can regard $L^{\otimes 2}$ as an ``$L$-module" via the $\alpha$-twisted adjoint action \eqref{eq:ad}:
\begin{equation}
\label{eq:L2}
\begin{split}
x \cdot (y_1 \otimes y_2)
&= \ad_{\alpha(x)}(y_1 \otimes y_2)\\
&= [\alpha(x),y_1] \otimes \alpha(y_2) + \alpha(y_1) \otimes [\alpha(x),y_2]
\end{split}
\end{equation}
for $x \in L$ and $y_1 \otimes y_2 \in L^{\otimes 2}$.  Then we can think of the cobracket $\Delta \colon L \to L^{\otimes 2}$ as a $1$-cochain $\Delta \in C^1(L,L^{\otimes 2})$ ($=$ the linear subspace of $\Hom(L,L^{\otimes 2})$ consisting of maps that commute with $\alpha$).  Generalizing \cite[Definition 5.3]{ms3} to include coefficients in $L^{\otimes 2}$, the differential on $\Delta$ is given by
\begin{equation}
\label{eq:1cobound}
\begin{split}
(\delta^1_{HL}\Delta)(x,y)
&= \Delta([x,y]) - x \cdot \Delta(y) + y \cdot \Delta(x)\\
&= \Delta([x,y]) - \ad_{\alpha(x)}(\Delta(y)) + \ad_{\alpha(y)}(\Delta(x)).
\end{split}
\end{equation}
So \eqref{eq:compatibility} says exactly that $\Delta \in C^1(L,L^{\otimes 2})$ is a $1$-cocycle.
\end{remark}

The following result shows that a Lie bialgebra deforms into a Hom-Lie bialgebra along any Lie bialgebra morphism.  It gives an efficient method for constructing Hom-Lie bialgebras.


\begin{theorem}
\label{thm:hlbdeform}
Let $(L,[-,-],\Delta)$ be a Lie bialgebra and $\alpha \colon L \to L$ be a Lie bialgebra morphism.  Then $L_\alpha = (L,[-,-]_\alpha,\Delta_\alpha,\alpha)$ is a Hom-Lie bialgebra, where $[-,-]_\alpha = \alpha \circ [-,-]$ and $\Delta_\alpha = \Delta \circ \alpha$.
\end{theorem}

\begin{proof}
As we mentioned earlier, it is easy to check directly that $(L,[-,-]_\alpha,\alpha)$ is a Hom-Lie algebra \cite{yau2}.  A similar argument shows that $(L,\Delta_\alpha,\alpha)$ is a Hom-Lie coalgebra.  In fact, the Hom-co-Jacobi identity for $\Delta_\alpha$ follows from that for $\Delta$ and the facts that $(\alpha \otimes \Delta_\alpha) \circ \Delta_\alpha = (\alpha^{\otimes 3})^2 \circ (Id \otimes \Delta) \circ \Delta$ and that $\alpha^{\otimes 3}$ commutes with any permutation.

It remains to prove the compatibility condition \eqref{eq:compatibility} for $\Delta_\alpha$ and $[-,-]_\alpha$.  In this case, the left-hand side of \eqref{eq:compatibility} is
\begin{equation}
\label{eq:compatLHS}
\Delta_\alpha([x,y]_\alpha) = \Delta \circ \alpha^2 \circ [x,y] = (\alpha^{\otimes 2})^2(\Delta([x,y])),
\end{equation}
since $\Delta \circ \alpha = \alpha^{\otimes 2} \circ \Delta$.  Using in addition $\alpha \circ [-,-] = [-,-] \circ \alpha^{\otimes 2}$, we have
\begin{equation}
\label{eq:compatRHS}
\begin{split}
\ad_{\alpha(x)}(\Delta_\alpha(y)) &=
[\alpha(x),\alpha(y)_1]_\alpha \otimes \alpha(\alpha(y)_2) + \alpha(\alpha(y)_1) \otimes [\alpha(x),\alpha(y)_2]_\alpha\\
&= (\alpha^{\otimes 2})^2\left([x,y_1] \otimes y_2 + y_1 \otimes [x,y_2]\right).
\end{split}
\end{equation}
It follows from \eqref{eq:compatLHS} and \eqref{eq:compatRHS} that
\[
\Delta_\alpha([x,y]_\alpha) = \ad_{\alpha(x)}(\Delta_\alpha(y)) - \ad_{\alpha(y)}(\Delta_\alpha(x)),
\]
since \eqref{eq:compatibility} with $\alpha = Id$ is assumed to hold.
\end{proof}

Next we consider when Hom-Lie bialgebras of the form $L_\alpha$, as in Theorem ~\ref{thm:hlbdeform}, are isomorphic.

\begin{proposition}
\label{prop:twist}
Let $\fg$ and $\fh$ be Lie bialgebras and $\alpha \colon \fg \to \fg$ and $\beta \colon \fh \to \fh$ be Lie bialgebra morphisms with $\beta$ and $\beta^{\otimes 2}$ injective.  Then the following statements are equivalent:
\begin{enumerate}
\item
The Hom-Lie bialgebras $\fg_\alpha$ and $\fh_\beta$ (as in Theorem ~\ref{thm:hlbdeform}) are isomorphic.
\item
There exists a Lie bialgebra isomorphism $\gamma \colon \fg \to \fh$ such that $\gamma\alpha = \beta\gamma$.
\end{enumerate}
\end{proposition}

\begin{proof}
To show that the first statement implies the second statement, suppose that $\gamma \colon \fg_\alpha \to \fh_\beta$ is an isomorphism of Hom-Lie bialgebras.  Then $\gamma\alpha = \beta\gamma$ automatically.  To see that $\gamma$ is a Lie bialgebra isomorphism, first we check that it commutes with the Lie brackets.  For any two elements $x$ and $y$ in $\fg$, we have
\[
\beta\gamma[x,y] = \gamma\alpha[x,y] = \gamma([x,y]_\alpha) = [\gamma(x),\gamma(y)]_\beta = \beta[\gamma(x),\gamma(y)].
\]
Since $\beta$ is injective, we conclude that $\gamma[x,y] = [\gamma(x),\gamma(y)]$, i.e., $\gamma$ is a Lie algebra isomorphism.

To check that $\gamma$ commutes with the Lie cobrackets, we compute as follows:
\[
\begin{split}
\beta^{\otimes 2}(\gamma^{\otimes 2}(\Delta(x))) &= (\beta\gamma)^{\otimes 2}(\Delta(x)) = (\gamma\alpha)^{\otimes 2}(\Delta(x)) = \gamma^{\otimes 2}(\alpha^{\otimes 2}(\Delta(x)))\\
&= \gamma^{\otimes 2}(\Delta_\alpha(x)) = \Delta_\beta(\gamma(x)) = \beta^{\otimes 2}(\Delta(\gamma(x))).
\end{split}
\]
The injectivity of $\beta^{\otimes 2}$ now implies that $\gamma$ commutes with the Lie cobrackets.  Therefore, $\gamma$ is a Lie bialgebra isomorphism.

The other implication is proved by a similar argument, much of which is already given above.
\end{proof}

For a Lie bialgebra $\fg$, let $\Aut(\fg)$ be the group of Lie bialgebra isomorphisms from $\fg$ to $\fg$.  In Proposition ~\ref{prop:twist}, restricting to the case $\fg = \fh$ with $\alpha$ and $\beta$ both invertible, we obtain the following special case.

\begin{corollary}
\label{cor:twist}
Let $\fg$ be a Lie bialgebra and $\alpha, \beta \in \Aut(\fg)$.  Then the Hom-Lie bialgebras $\fg_\alpha$ and $\fg_\beta$ (as in Theorem ~\ref{thm:hlbdeform}) are isomorphic if and only if $\alpha$ and $\beta$ are conjugate in $\Aut(\fg)$.
\end{corollary}

Corollary ~\ref{cor:twist} can be restated as follows.

\begin{corollary}
\label{cor2:twist}
Let $\fg$ be a Lie bialgebra.  Then there is a bijection between (i) the set of isomorphism classes of Hom-Lie bialgebras $\fg_\alpha$ with $\alpha$ invertible and (ii) the set of conjugacy classes in the group $\Aut(\fg)$.
\end{corollary}

As we will show later in this section, Corollary ~\ref{cor2:twist} implies that there are uncountably many isomorphism classes of Hom-Lie bialgebras of the form $sl(2)_\alpha$.

The next result shows that finite dimensional Hom-Lie bialgebras, like Lie bialgebras, can be dualized.  A proof of this self-dual property for the special case of Lie bialgebras can be found in \cite[Proposition 8.1.2]{majid}.


\begin{theorem}
\label{thm;duality}
Let $(L,[-,-],\Delta,\alpha)$ be a finite dimensional Hom-Lie bialgebra.  Then its linear dual $L^* = \Hom(L,\bk)$ is also a Hom-Lie bialgebra with the dual structure maps:
\begin{equation}
\label{eq:dualmaps}
\alpha(\phi) = \phi \circ \alpha, \quad
\langle [\phi,\psi],x \rangle = \langle \phi \otimes \psi,\Delta(x)\rangle,\quad
\langle \Delta(\phi), x \otimes y\rangle = \langle \phi,[x,y]\rangle
\end{equation}
for $x, y \in L$ and $\phi, \psi \in L^*$.
\end{theorem}

\begin{proof}
As we mentioned right after Definition ~\ref{def:hlc}, $(L^*,[-,-],\alpha)$ is a Hom-Lie algebra (which is true even if $L$ is not finite dimensional) and $(L^*,\Delta,\alpha)$ is a Hom-Lie coalgebra (whose validity depends on the finite dimensionality of $L$).  Thus, it remains to check the compatibility condition \eqref{eq:compatibility} between the bracket and the cobracket in $L^*$, i.e.,
\begin{equation}
\label{eq:dualcompat}
\langle \Delta[\phi,\psi], x \otimes y\rangle = \langle \ad_{\alpha(\phi)}(\Delta\psi) - \ad_{\alpha(\psi)}(\Delta\phi), x \otimes y\rangle
\end{equation}
for $x, y \in L$ and $\phi, \psi \in L^*$.

Using the definitions \eqref{eq:dualmaps} and the compatibility condition \eqref{eq:compatibility} (and its expanded form \eqref{eq:compat}) in $L$, we compute the left-hand side of \eqref{eq:dualcompat} as follows:
\begin{align*}
\allowdisplaybreaks
\langle \Delta[\phi,\psi],  x \otimes y\rangle &= \langle [\phi,\psi],[x,y]\rangle = \langle \phi\otimes\psi, \Delta[x,y]\rangle\rangle
= \langle \phi\otimes\psi, \ad_{\alpha(x)}(\Delta(y)) - \ad_{\alpha(y)}(\Delta(x))\rangle\\
&= \langle\phi\otimes\psi,[\alpha(x),y_1]\otimes\alpha(y_2)\rangle +
\langle\phi\otimes\psi,\alpha(y_1)\otimes[\alpha(x),y_2]\rangle\\
&\relphantom{} - \langle\phi\otimes\psi,[\alpha(y),x_1]\otimes\alpha(x_2)\rangle -
\langle\phi\otimes\psi,\alpha(x_1)\otimes[\alpha(y),x_2]\rangle\\
&= \langle \phi_1 \otimes \phi_2 \otimes \psi,\alpha(x)\otimes y_1 \otimes \alpha(y_2)\rangle +
\langle \phi \otimes \psi_1 \otimes \psi_2,\alpha(y_1) \otimes \alpha(x) \otimes y_2\rangle\\
&\relphantom{} - \langle \phi_1 \otimes \phi_2 \otimes \psi,\alpha(y) \otimes x_1 \otimes \alpha(x_2)\rangle -
\langle \phi \otimes \psi_1 \otimes \psi_2,\alpha(x_1) \otimes \alpha(y) \otimes x_2\rangle.\\
&= \langle \alpha(\phi_1) \otimes \phi_2 \otimes \alpha(\psi),x\otimes y_1 \otimes y_2\rangle +
\langle \alpha(\phi) \otimes \alpha(\psi_1) \otimes \psi_2,y_1 \otimes x \otimes y_2\rangle\\
&\relphantom{} - \langle \alpha(\phi_1) \otimes \phi_2 \otimes \alpha(\psi),y \otimes x_1 \otimes x_2\rangle -
\langle \alpha(\phi) \otimes \alpha(\psi_1) \otimes \psi_2,x_1 \otimes y \otimes x_2\rangle.\\
\intertext{Using, in addition, the anti-symmetry of the bracket and the cobracket in $L^*$, the above four terms become:}
&= -\langle \alpha(\phi_1) \otimes [\alpha(\psi),\phi_2],x \otimes y\rangle +
\langle \alpha(\psi_1) \otimes [\alpha(\phi),\psi_2],x \otimes y\rangle\\
&\relphantom{} - \langle [\alpha(\psi),\phi_1] \otimes \alpha(\phi_2),x \otimes y\rangle +
\langle [\alpha(\phi),\psi_1] \otimes \alpha(\psi_2),x \otimes y\rangle.
\end{align*}
This is exactly the right-hand side of \eqref{eq:dualcompat} in expanded form \eqref{eq:compat}, as desired.
\end{proof}


Let us illustrate the results in this section with the Lie bialgebra $sl(2)$.  As we recalled in the previous section, this complex Lie algebra has a basis $\{H,X_{\pm}\}$ \eqref{slbracket}.  It becomes a Lie bialgebra when equipped with the cobracket $\Delta \colon sl(2) \to sl(2)^{\otimes 2}$ \cite{dri83,dri87} (see also \cite[Example 8.1.10]{majid}) defined as
\begin{equation}
\label{slcobracket}
\Delta(H) = 0 \quad\text{and}\quad
\Delta(X_{\pm}) = \frac{1}{2}\left(X_{\pm} \otimes H - H \otimes X_{\pm}\right).
\end{equation}
We will construct all the Hom-Lie bialgebras of the form $sl(2)_\alpha$ using Theorem ~\ref{thm:hlbdeform}.  First we compute the Lie bialgebra morphisms on $sl(2)$.

\begin{proposition}
\label{prop:slbimorphism}
With respect to the basis $\{H,X_{\pm}\}$, a non-zero linear map $\alpha \colon sl(2) \to sl(2)$ is a Lie bialgebra morphism if and only if $\alpha(H) = H$ and $\alpha(X_{\pm}) = b^{\pm 1}X_{\pm}$ for some non-zero complex number $b$.
\end{proposition}

\begin{proof}
In the previous section, we recalled the classification of Lie algebra morphisms on $sl(2)$ \cite{yau6}.  The non-zero Lie algebra morphisms on $sl(2)$ must have one of the three forms: $\alpha_1$ \eqref{alpha1}, $\alpha_2$ \eqref{alpha2}, or $\alpha_3$ \eqref{alpha3}.  In the context of this classification, the Proposition is equivalent to saying that the Lie bialgebra morphisms on $sl(2)$ are exactly the $\alpha_1$ with $a = c = 0$.  It is immediate that $\alpha_1$ with $a = c = 0$ commutes with the cobracket $\Delta$ \eqref{slcobracket} and is, therefore, a Lie bialgebra morphism.  It remains to check that they are the only non-zero Lie bialgebra morphisms on $sl(2)$.

If $\alpha = \alpha_1$ \eqref{alpha1} is a Lie bialgebra morphism on $sl(2)$, then we have
\[
\begin{split}
0 &= \alpha^{\otimes 2}(\Delta(H)) = \Delta(\alpha(H)) = \Delta(H - 2abX_+ - 2b^{-1}cX_-)\\
&= -ab|X_+ \otimes H| - b^{-1}c|X_- \otimes H|,
\end{split}
\]
where the abbreviation \eqref{absolute} is used.  The above element in $sl(2)^{\otimes 2}$ is $0$ if and only if $a = c = 0$.  Next we show that maps of the forms $\alpha_2$ and $\alpha_3$ are not Lie bialgebra morphisms on $sl(2)$.

If $\alpha = \alpha_2$ \eqref{alpha2} is a Lie bialgebra morphism on $sl(2)$, then a similar computation as in the previous paragraph implies $a = c = 0$.  In other words, we must have $\alpha(H) = -H$, $\alpha(X_+) = bX_-$, and $\alpha(X_-) = b^{-1}X_+$.  In this case, on the one hand, we have
\begin{equation}
\label{alphadelta1}
\alpha^{\otimes 2}(\Delta(X_+)) = \frac{1}{2}|\alpha(X_+) \otimes \alpha(H)| = -\frac{b}{2}|X_-\otimes H|.
\end{equation}
On the other hand, we have
\begin{equation}
\label{alphadelta2}
\Delta(\alpha(X_+)) = b\Delta(X_-) = \frac{b}{2}|X_-\otimes H|.
\end{equation}
The equality between \eqref{alphadelta1} and \eqref{alphadelta2} then implies $b = 0$, which is a contradiction.  Therefore, maps of the form $\alpha_2$ are not Lie bialgebra morphisms on $sl(2)$.

Finally, suppose that $\alpha = \alpha_3$ \eqref{alpha3} is a Lie bialgebra morphism on $sl(2)$.  Then a similar computation as above, applied to $\alpha^{\otimes 2}(\Delta(H)) = 0 = \Delta(\alpha(H))$, implies $b = 0$.  This is again a contradiction.  Therefore, maps of the form $\alpha_3$ are not Lie bialgebra morphisms on $sl(2)$.
\end{proof}

Combining Theorem ~\ref{thm:hlbdeform}, Proposition ~\ref{prop:slbimorphism}, and the definitions \eqref{slbracket} and \eqref{slcobracket} of the (co)bracket in $sl(2)$, we obtain the following family of Hom-Lie bialgebras.

\begin{corollary}
\label{cor1:HLsl}
Suppose that $\alpha \colon sl(2) \to sl(2)$ is the Lie bialgebra morphism given by $\alpha(H) = H$ and $\alpha(X_{\pm}) = b^{\pm 1}X_{\pm}$ for some non-zero complex number $b$.  Then there is a Hom-Lie bialgebra $sl(2)_\alpha = (sl(2),[-,-]_\alpha,\Delta_\alpha,\alpha)$, in which the bracket and the cobracket are determined by
\begin{equation}
\label{slstructure}
\begin{split}
[H,X_{\pm}]_\alpha &= \pm 2b^{\pm 1}X_{\pm}, \quad [X_+,X_-]_\alpha = H,\\
\Delta_\alpha(H) &= 0,\quad \Delta_\alpha(X_{\pm}) = \frac{b^{\pm 1}}{2}(X_{\pm} \otimes H - H \otimes X_{\pm}).
\end{split}
\end{equation}
\end{corollary}

As we will see in Corollary ~\ref{cor:sl2quasi} in the next section, the Hom-Lie bialgebras $sl(2)_\alpha$ have the additional property of being \emph{quasi-triangular} (Definition ~\ref{def:coboundary}).  This means that the classical $r$-matrix $r \in sl(2)^{\otimes 2}$ \eqref{eq:slr} is fixed by $\alpha^{\otimes 2}$, that it induces the cobracket $\Delta_\alpha$ via the adjoint map $\ad(r)$ \eqref{eq:ad}, and that $r$ is a solution of the CHYBE.

Proposition ~\ref{prop:slbimorphism} also tells us that the group $\Aut(sl(2))$ of Lie bialgebra isomorphisms on $sl(2)$ is isomorphic to $\bC^*$, the multiplicative group of non-zero complex numbers.  In particular, it is an abelian group, and so two elements in it are conjugate if and only if they are equal.  Combining Corollary ~\ref{cor:twist} and Corollary ~\ref{cor1:HLsl}, we have the following result, which implies that there are uncountably many non-isomorphic Hom-Lie bialgebras of the form $sl(2)_\alpha$.

\begin{corollary}
\label{cor2:HLsl}
Two Hom-Lie bialgebras of the form $sl(2)_\alpha$ (as in Corollary ~\ref{cor1:HLsl}) are isomorphic if and only if the associated scalars $b$ are equal.  In particular, there is a bijection between the set of isomorphism classes of Hom-Lie bialgebras of the form $sl(2)_\alpha$ and the set of non-zero complex numbers.
\end{corollary}

Next we describe the dual $sl(2)_\alpha^* = \Hom(sl(2)_\alpha,\bC)$ of the three-dimensional Hom-Lie bialgebra $sl(2)_\alpha$ in Corollary ~\ref{cor1:HLsl}.  Let $\{\phi,\psi_+,\psi_-\}$ be the dual basis of $sl(2)_\alpha^*$.  In other words, these basis elements are determined by
\[
\langle \phi,H\rangle = 1 = \langle \psi_{\pm},X_{\pm}\rangle,
\]
where $\{H,X_{\pm}\}$ is the standard basis of $sl(2)$ \eqref{slbracket}.  The following result is a generalization of \cite[Example 8.1.11]{majid}, which describes the dual Lie bialgebra $sl(2)^*$.

\begin{corollary}
\label{cor:dual}
Let $sl(2)_\alpha$ be the Hom-Lie bialgebra in Corollary ~\ref{cor1:HLsl}.  Then the structure maps of its dual Hom-Lie bialgebra $sl(2)_\alpha^*$, in the sense of Theorem ~\ref{thm;duality}, are determined by:
\begin{equation}
\label{sldualstructure}
\begin{split}
\alpha(\phi) &= \phi \circ \alpha, \quad \alpha(\psi_{\pm}) = \psi_{\pm} \circ \alpha,\quad
[\psi_{\pm},\phi]_\alpha = \frac{b^{\pm 1}}{2}\psi_{\pm}, \quad [\psi_+,\psi_-]_\alpha = 0,\\
\Delta_\alpha(\psi_{\pm}) &= \pm 2b^{\pm 1}(\phi \otimes \psi_{\pm} - \psi_{\pm} \otimes \phi),\quad
\Delta_\alpha(\phi) = \psi_+ \otimes \psi_- - \psi_- \otimes \psi_+.
\end{split}
\end{equation}
\end{corollary}

\begin{proof}
One can check directly that \eqref{sldualstructure} defines a Hom-Lie bialgebra structure on $sl(2)_\alpha^*$.  It remains to check that the bracket and the cobracket in $sl(2)_\alpha^*$ are dual to, respectively, the cobracket and the bracket \eqref{slstructure} in $sl(2)_\alpha$, in the sense of \eqref{eq:dualmaps}.  Using the anti-symmetry of these (co)brackets, we only need to check these duality properties for the determining brackets on the basis elements.  Most of these equalities hold because both sides are zero.  The remaining non-trivial ones are:
\[
\begin{split}
\langle [\psi_{\pm},\phi]_\alpha, X_{\pm}\rangle &= \frac{b^{\pm 1}}{2} = \langle \psi_{\pm} \otimes \phi, \Delta_\alpha(X_{\pm})\rangle,\qquad
\langle \Delta_\alpha(\phi),X_+\otimes X_-\rangle = 1 = \langle \phi,[X_+,X_-]_\alpha\rangle,\\
\langle \Delta_\alpha(\psi_{\pm}),H\otimes X_{\pm}\rangle &= \pm 2b^{\pm 1} = \langle \psi_{\pm},[H,X_{\pm}]_\alpha\rangle.
\end{split}
\]
\end{proof}

\section{Coboundary and quasi-triangular Hom-Lie bialgebras}
\label{sec;cob}

The connections between the CHYBE \eqref{eq:chybe} and Hom-Lie bialgebras (Definition ~\ref{def:hlb}) arise in the sub-classes of coboundary and quasi-triangular Hom-Lie bialgebras.  We first prove the analogue of Theorem ~\ref{thm:hlbdeform} for coboundary/quasi-triangular Hom-Lie bialgebras  (Theorem ~\ref{thm:cob}), which  gives an efficient method for constructing these objects from coboundary/quasi-triangular Lie bialgebras.  As an example, we observe that the Hom-Lie bialgebras $sl(2)_\alpha$ in Corollary \ref{cor1:HLsl} are all quasi-triangular (Corollary ~\ref{cor:sl2quasi}).  Then we show how a coboundary/quasi-triangular Hom-Lie bialgebra can be constructed from a Hom-Lie algebra and a suitable element $r \in L^{\otimes 2}$ (Theorem ~\ref{thm:cob-char} and Corollary ~\ref{cor:cob}).  This section ends with several characterizations of when a coboundary Hom-Lie bialgebra is a quasi-triangular Hom-Lie bialgebra (Theorem ~\ref{thm:quasi}).

Recall the adjoint map in \eqref{eq:ad}.  Here are the relevant definitions.


\begin{definition}
\label{def:coboundary}
A \textbf{coboundary Hom-Lie bialgebra} $(L,[-,-],\Delta,\alpha,r)$ consists of a Hom-Lie bialgebra $(L,[-,-],\Delta,\alpha)$ and an element $r = \sum r_1 \otimes r_2 \in L^{\otimes 2}$ such that $\alpha^{\otimes 2}(r) = r$ and
\begin{equation}
\label{eq:Delta=ad}
\Delta(x) = \ad_x(r) = \sum [x,r_1] \otimes \alpha(r_2) + \alpha(r_1) \otimes [x,r_2]
\end{equation}
for all $x \in L$.  A \textbf{quasi-triangular Hom-Lie bialgebra} is a coboundary Hom-Lie bialgebra in which $r$ is a solution of the CHYBE \eqref{eq:chybe}.  In these cases, we also write $\Delta$ as $\ad(r)$.
\end{definition}

A coboundary/quasi-triangular Hom-Lie bialgebra in which $\alpha = Id$ is exactly a coboundary/quasi-triangular Lie bialgebra, as defined by Drinfel'd \cite{dri83,dri87}.  One can also use this as the definition of a coboundary/quasi-triangular Lie bialgebra, which we denote by $(L,[-,-],\Delta,r)$.  To be more precise, a \emph{coboundary Lie bialgebra} is a Lie bialgebra in which the Lie cobracket $\Delta$ takes the form \eqref{eq:Delta=ad} with $\alpha = Id$.  A \emph{quasi-triangular Lie bialgebra} is a coboundary Lie bialgebra in which $r$ is a solution of the CYBE, which is the CHYBE \eqref{eq:chybe} with $\alpha = Id$.  Note that we do not require $r$ to be anti-symmetric in a coboundary Hom-Lie bialgebra, whereas in \cite{dri87} $r$ is assumed to be anti-symmetric in a coboundary Lie bialgebra.  Our convention follows that of \cite{majid}.

\begin{remark}
\label{rk:delta}
Let us explain why \eqref{eq:Delta=ad} is a natural condition.  Recall from Remark \ref{rk:cocycle} that the compatibility condition \eqref{eq:compatibility} in a Hom-Lie bialgebra $L$ says that the cobracket $\Delta$ is a $1$-cocycle in $C^1(L,L^{\otimes 2})$, where $L$ acts on $L^{\otimes 2}$ via the $\alpha$-twisted adjoint action \eqref{eq:L2}.  The simplest $1$-cocycles are the $1$-coboundaries, i.e., images of $\delta^0_{HL}$.  We can define the Hom-Lie $0$-cochains and $0$th differential as follows, extending the definitions in \cite[section 5]{ms3}.  Set $C^0(L,L^{\otimes 2})$ as the subspace of $L^{\otimes 2}$ consisting of elements that are fixed by $\alpha^{\otimes 2}$.  Then we define the differential $\delta^0_{HL} \colon C^0(L,L^{\otimes 2}) \to C^1(L,L^{\otimes 2})$ by setting $\delta^0_{HL}(r) = \ad(r)$, as in \eqref{eq:ad}.  It is not hard to check that, for $r \in C^0(L,L^{\otimes 2})$, we have $\delta^1_{HL}(\delta^0_{HL}(r)) = 0$, where $\delta^1_{HL}$ is defined in \eqref{eq:1cobound}.  In fact, what this condition says is that
\begin{equation}
\label{compose}
\begin{split}
0 &= \delta^1_{HL}(\delta^0_{HL}(r))(x,y) = \delta^1_{HL}(\ad(r))(x,y)\\
&= \ad_{[x,y]}(r) - \ad_{\alpha(x)}(\ad_y(r)) + \ad_{\alpha(y)}(\ad_x(r))
\end{split}
\end{equation}
for all $x,y \in L$.  We will prove \eqref{compose} in Lemma ~\ref{lem:perturb} below.  So such a $\delta^0_{HL}(r) = \ad(r)$ is a $1$-coboundary, and hence a $1$-cocycle.  This fact makes $\ad(r)$ (with $\alpha^{\otimes 2}(r) = r$) a natural candidate for the cobracket in a Hom-Lie bialgebra and also justifies the name \emph{coboundary Hom-Lie bialgebra}.
\end{remark}

The following result is the analogue of Theorem ~\ref{thm:hlbdeform} for coboundary/quasi-triangular Hom-Lie bialgebras.  It says that these objects can be obtained by twisting coboundary/quasi-triangular Lie bialgebras via a suitable morphism.

\begin{theorem}
\label{thm:cob}
Let $(L,[-,-],\Delta,r)$ be a coboundary Lie bialgebra and $\alpha \colon L \to L$ be a Lie algebra morphism such that $\alpha^{\otimes 2}(r) = r$.  Then the following two statements hold.
\begin{enumerate}
\item
$L_\alpha = (L,[-,-]_\alpha, \Delta_\alpha,\alpha,r)$ is a coboundary Hom-Lie bialgebra, where $[-,-]_\alpha = \alpha \circ [-,-]$ and $\Delta_\alpha = \Delta \circ \alpha$.
\item
If, in addition, $L$ is a quasi-triangular Lie bialgebra, then $L_\alpha$ is a quasi-triangular Hom-Lie bialgebra.
\end{enumerate}
\end{theorem}

\begin{proof}
For the first statement, we need to show two things: (i) $L_\alpha = (L,[-,-]_\alpha,\Delta_\alpha,\alpha)$ is a Hom-Lie bialgebra, and (ii) the condition \eqref{eq:Delta=ad} holds for $\Delta_\alpha$ and $[-,-]_\alpha$.  For (i), we use Theorem ~\ref{thm:hlbdeform}.  Since $L$ is a Lie bialgebra and $\alpha$ is a Lie algebra morphism, it remains to show that $\alpha$ is compatible with the Lie cobracket $\Delta = \ad(r)$ as well.  Write $r = \sum r_1 \otimes r_2$.  Using $\alpha[-,-] = [-,-]\circ \alpha^{\otimes 2}$ and $\alpha^{\otimes 2}(r) = r$, we compute as follows:
\[
\begin{split}
\alpha^{\otimes 2}(\Delta(x)) &=
\alpha^{\otimes 2}\left([x,r_1] \otimes r_2 + r_1 \otimes [x,r_2]\right)\\
&= [\alpha(x),\alpha(r_1)] \otimes \alpha(r_2) + \alpha(r_1) \otimes [\alpha(x),\alpha(r_2)]\\
&= [\alpha(x),r_1] \otimes r_2 + r_1 \otimes [\alpha(x),r_2] = \Delta(\alpha(x)).
\end{split}
\]
Therefore, $\alpha$ is a Lie bialgebra morphism on $L$.  It follows from Theorem ~\ref{thm:hlbdeform} that $L_\alpha$ is a Hom-Lie bialgebra, thereby proving (i).

For (ii), we compute similarly as follows:
\[
\begin{split}
\Delta_\alpha(x) &= \Delta(\alpha(x)) =  [\alpha(x),\alpha(r_1)] \otimes \alpha(r_2) + \alpha(r_1) \otimes [\alpha(x),\alpha(r_2)]\\
&= [x,r_1]_\alpha \otimes \alpha(r_2) + \alpha(r_1) \otimes [x,r_2]_\alpha.
\end{split}
\]
This shows that $\Delta_\alpha$ and $[-,-]_\alpha$ satisfy \eqref{eq:Delta=ad}.  We have shown that $L_\alpha$ is a coboundary Hom-Lie bialgebra.

For statement (2), we assume, in addition, that $r$ is a solution of the CYBE (i.e., \eqref{eq:chybe} with $\alpha = Id$).  By Theorem ~\ref{thm:chybe} (the case $n = 0$), we know that $r$ is also a solution of the CHYBE in $L_\alpha$.  Together with part (1), we conclude that $L_\alpha$ is a quasi-triangular Hom-Lie bialgebra.
\end{proof}

The following result is an illustration of Theorem ~\ref{thm:cob}.

\begin{corollary}
\label{cor:sl2quasi}
The Hom-Lie bialgebras $sl(2)_\alpha$ in Corollary ~\ref{cor1:HLsl} are all quasi-triangular with $r = X_+ \otimes X_- + \frac{1}{4}H \otimes H$ \eqref{eq:slr}.
\end{corollary}

\begin{proof}
It is known that $sl(2)$ is a quasi-triangular Lie bialgebra \cite{dri83,dri87} (or \cite[Example 8.1.10]{majid}) with the Lie cobracket $\Delta$ \eqref{slcobracket} and the classical $r$-matrix $r = X_+ \otimes X_- + \frac{1}{4}H \otimes H$ \eqref{eq:slr}.  In Corollary ~\ref{cor1:HLsl}, the maps $\alpha \colon sl(2) \to sl(2)$ are Lie bialgebra morphisms (computed in Proposition ~\ref{prop:slbimorphism}) of the form $\alpha(H) = H$ and $\alpha(X_{\pm}) = b^{\pm 1}X_{\pm}$ for some non-zero scalar $b$.  By the second part of Theorem ~\ref{thm:cob}, to show that $sl(2)_\alpha$ is a quasi-triangular Hom-Lie bialgebra, it remains to show $\alpha^{\otimes 2}(r) = r$.  This is easy, since
\[
\begin{split}
\alpha^{\otimes 2}(r) &= \alpha(X_+) \otimes \alpha(X_-) + \frac{1}{4}\alpha(H) \otimes \alpha(H)\\
& = (bX_+) \otimes (b^{-1}X_-) + \frac{1}{4}H \otimes H = r.
\end{split}
\]
\end{proof}

In fact, the only Lie algebra morphisms on $sl(2)$ that fix $r  = X_+ \otimes X_- + \frac{1}{4}H \otimes H$ (i.e., $\alpha^{\otimes 2}(r) = r$) are the non-zero Lie bialgebra morphisms (Proposition ~\ref{prop:slbimorphism}).  This follows from the classification of non-zero Lie algebra morphisms on $sl(2)$ (\eqref{alpha1} - \eqref{alpha3}) and Propositions \ref{prop:sl1} - \ref{prop:sl3} (the case $n = 1$).


In the following result, we describe some sufficient conditions under which a Hom-Lie algebra becomes a coboundary Hom-Lie bialgebra.  It is a generalization of \cite[Proposition 8.1.3]{majid}, which deals with Lie algebras and coboundary Lie bialgebras.  In what follows, for an element $r = \sum r_1 \otimes r_2$, we write $r_{21}$ for $\tau(r) = \sum r_2 \otimes r_1$.

\begin{theorem}
\label{thm:cob-char}
Let $(L,[-,-],\alpha)$ be a Hom-Lie algebra and $r \in L^{\otimes 2}$ be an element such that $\alpha^{\otimes 2}(r) = r$, $r_{21} = -r$, and
\begin{equation}
\label{eq:cobound}
\alpha^{\otimes 3}(\ad_x([[r,r]]^\alpha)) = 0
\end{equation}
for all $x \in L$, where $[[r,r]]^\alpha$ is defined in \eqref{eq:chybe}.  Define $\Delta \colon L \to L^{\otimes 2}$ by $\Delta(x) = \ad_x(r)$ as in \eqref{eq:Delta=ad}.  Then $(L,[-,-],\Delta,\alpha,r)$ is a coboundary Hom-Lie bialgebra.
\end{theorem}

\begin{proof}
We will show that (i) $\Delta = \ad(r)$ commutes with $\alpha$, (ii) $\Delta$ is anti-symmetric, (iii) the compatibility condition \eqref{eq:compatibility} holds, and (iv) the condition \eqref{eq:cobound} is equivalent to the Hom-co-Jacobi identity of $\Delta$.

Write $r$ as $\sum r_1 \otimes r_2$.  To show that $\Delta = \ad(r)$ commutes with $\alpha$, pick an element $x \in L$.  Then $\Delta(\alpha(x))$ and $\alpha^{\otimes 2}(\Delta(x))$ are both equal to $\alpha([x,r_1]) \otimes \alpha^2(r_2) + \alpha^2(r_1) \otimes \alpha([x,r_2])$.  This follows from the definition $\Delta = \ad(r)$, $\alpha([-,-]) = [-,-] \circ \alpha^{\otimes 2}$, and the assumption $\alpha^{\otimes 2}(r) = r$.

Now we show that $\Delta = \ad(r)$ is anti-symmetric.  We have
\begin{equation}
\label{deltaantisym}
\allowdisplaybreaks
\begin{split}
\Delta(x) + \tau(\Delta(x)) &= [x,r_1] \otimes \alpha(r_2) + \alpha(r_1) \otimes [x,r_2] + \alpha(r_2) \otimes [x,r_1] + [x,r_2] \otimes \alpha(r_1)\\
&= \ad_x(r + r_{21}) = \ad_x(0) = 0,
\end{split}
\end{equation}
since $r + r_{21} = \sum (r_1 \otimes r_2 + r_2 \otimes r_1)$.

We will prove that the compatibility condition \eqref{eq:compatibility} holds in Lemma ~\ref{lem:perturb} below.

Finally, we show that the Hom-co-Jacobi identity (Definition ~\ref{def:hlc}) of $\Delta = \ad(r)$ is equivalent to \eqref{eq:cobound}.  Let us unwrap the Hom-co-Jacobi identity.  Fix an element $x \in L$, and let $r' = \sum r_1' \otimes r_2'$ be another copy of $r$.  Then we write
\[
\allowdisplaybreaks
\begin{split}
\gamma &= (\alpha \otimes \Delta)(\Delta(x))
= (\alpha \otimes \Delta)([x,r_1] \otimes \alpha(r_2) + \alpha(r_1) \otimes [x,r_2])\\
&= \alpha([x,r_1]) \otimes [\alpha(r_2),r_1'] \otimes \alpha(r_2') + \alpha([x,r_1]) \otimes \alpha(r_1') \otimes [\alpha(r_2),r_2']\\
&\relphantom{} + \alpha^2(r_1) \otimes [[x,r_2],r_1'] \otimes \alpha(r_2') + \alpha^2(r_1) \otimes \alpha(r_1') \otimes [[x,r_2],r_2']\\
&= A_1 + B_1 + C_1 + D_1.
\end{split}
\]
In the last line above, we defined $A_1$ as $\alpha([x,r_1]) \otimes [\alpha(r_2),r_1'] \otimes \alpha(r_2')$, $B_1$ as $\alpha([x,r_1]) \otimes \alpha(r_1') \otimes [\alpha(r_2),r_2']$, and so forth.  Recall from \ref{notations} that $\sigma$ is the cyclic permutation given by $\sigma(1) = 2$, $\sigma(2) = 3$, and $\sigma(3) = 1$.  Applying $\sigma$ and $\sigma^2$ to $\gamma$ above, we obtain four similar but permutated terms in each case.  As above, we write $\sigma(\gamma)$ as $A_2 + B_2 + C_2 + D_2$ and $\sigma^2(\gamma)$ as $A_3 + B_3 + C_3 + D_3$, where $A_2 = \sigma(A_1) = \alpha(r_2') \otimes \alpha([x,r_1]) \otimes [\alpha(r_2),r_1']$, $A_3 = \sigma^2(A_1) = [\alpha(r_2),r_1'] \otimes \alpha(r_2') \otimes \alpha([x,r_1])$, and similarly for $B_i$, $C_i$, and $D_i$ ($i = 2, 3$).  With these notations, the Hom-co-Jacobi identity of $\Delta = \ad(r)$ (applied to $x$) becomes
\begin{equation}
\label{cojacobi}
0 = \cyclicsum (\alpha \otimes \Delta)(\Delta(x)) = \gamma + \sigma(\gamma) + \sigma^2(\gamma)
= \sum_{i=1}^3 (A_i + B_i + C_i + D_i).
\end{equation}
Therefore, to prove the equivalence between the Hom-co-Jacobi identity of $\Delta$ and \eqref{eq:cobound}, it suffices to show
\begin{equation}
\label{cojacobiad}
\alpha^{\otimes 3}(\ad_x([[r,r]]^\alpha)) = \sum_{i=1}^3 (A_i + B_i + C_i + D_i),
\end{equation}
which we will prove in Lemma ~\ref{lem:ABCD} below.

The proof of Theorem ~\ref{thm:cob-char} will be complete once we prove the two Lemmas below.
\end{proof}

\begin{lemma}
\label{lem:perturb}
Let $(L,[-,-],\alpha)$ be a Hom-Lie algebra and $r \in L^{\otimes 2}$ be an element such that $\alpha^{\otimes 2}(r) = r$.  Then $\Delta = \ad(r) \colon L \to L^{\otimes 2}$ satisfies \eqref{eq:compatibility}, i.e.,
\[
\ad_{[x,y]}(r) = \ad_{\alpha(x)}(\ad_y(r)) - \ad_{\alpha(y)}(\ad_x(r))
\]
for $x,y \in L$.
\end{lemma}

\begin{proof}
We will use $\alpha^{\otimes 2}(r) = r$, the anti-symmetry and the Hom-Jacobi identity of $[-,-]$ \eqref{eq:HomJacobi}, and $\alpha([-,-]) = [-,-]\circ\alpha^{\otimes 2}$ in the computation below.  For $x, y \in L$, we have:
\[
\allowdisplaybreaks
\begin{split}
\ad_{[x,y]}(r) &= [[x,y],r_1] \otimes \alpha(r_2) + \alpha(r_1) \otimes [[x,y],r_2]\\
&= [[x,y],\alpha(r_1)] \otimes \alpha^2(r_2) + \alpha^2(r_1) \otimes [[x,y],\alpha(r_2)]\\
&= \left\{[\alpha(x),[y,r_1]] + [\alpha(y),[r_1,x]]\right\} \otimes \alpha^2(r_2) + \alpha^2(r_1) \otimes \left\{[\alpha(x),[y,r_2]] + [\alpha(y),[r_2,x]]\right\}\\
&= [\alpha(x),[y,r_1]] \otimes \alpha^2(r_2) + \alpha([y,r_1]) \otimes [\alpha(x),\alpha(r_2)] + [\alpha(x),\alpha(r_1)] \otimes \alpha([y,r_2])\\
&\relphantom{} + \alpha^2(r_1) \otimes [\alpha(x),[y,r_2]] - [\alpha(y),[x,r_1]] \otimes \alpha^2(r_2) - \alpha([x,r_1]) \otimes [\alpha(y),\alpha(r_2)]\\
&\relphantom{} - [\alpha(y),\alpha(r_1)]\otimes \alpha([x,r_2]) - \alpha^2(r_1) \otimes [\alpha(y),[x,r_2]]\\
&= \ad_{\alpha(x)}\left([y,r_1] \otimes \alpha(r_2) + \alpha(r_1) \otimes [y,r_2]\right) - \ad_{\alpha(y)}\left([x,r_1] \otimes \alpha(r_2) + \alpha(r_1) \otimes [x,r_2]\right)\\
&= \ad_{\alpha(x)}(\ad_{y}(r)) - \ad_{\alpha(y)}(\ad_{x}(r)).
\end{split}
\]
In the fourth equality above, we added four terms (those not of the forms $\alpha^2(r_1) \otimes (\cdots)$ and $(\cdots) \otimes \alpha^2(r_2)$), which add up to zero.  Thus, the compatibility condition \eqref{eq:compatibility} holds.
\end{proof}

\begin{lemma}
\label{lem:ABCD}
The condition ~\eqref{cojacobiad} holds.
\end{lemma}

\begin{proof}
It suffices to show the following three equalities:
\begin{subequations}
\begin{align}
\alpha^{\otimes 3}(\ad_x([r_{12},r_{13}])) &= A_3 + B_2 + C_3 + D_2,\label{ad1213}\\
\alpha^{\otimes 3}(\ad_x([r_{12},r_{23}])) &= A_1 + B_3 + C_1 + D_3,\label{ad1223}\\
\alpha^{\otimes 3}(\ad_x([r_{13},r_{23}])) &= A_2 + B_1 + C_2 + D_1,\label{ad1323}
\end{align}
\end{subequations}
where the three brackets, which add up to $[[r,r]]^\alpha$, are defined in \eqref{eq:r123}.  The proofs for the three equalities are very similar, so we will only give the proof of \eqref{ad1213}.  Since $r = r'$ and $r_{21} = -r$, we have
\begin{equation}
\label{A3}
\begin{split}
A_3 &= [\alpha(r_2),r_1'] \otimes \alpha(r_2') \otimes \alpha([x,r_1])\\
&= [\alpha(r_2'),r_1] \otimes \alpha(r_2) \otimes \alpha([x,r_1'])\\
&= -[\alpha(r_1'),r_1] \otimes \alpha(r_2) \otimes \alpha([x,r_2'])\\
&= -[\alpha^2(r_1'),\alpha^2(r_1)] \otimes \alpha^3(r_2) \otimes \alpha[x,\alpha(r_2')]\\
&= \alpha^{\otimes 3}\left\{\alpha([r_1,r_1']) \otimes \alpha^2(r_2) \otimes [x,\alpha(r_2')]\right\}.
\end{split}
\end{equation}
In the fourth equality we used $(\alpha^{\otimes 2})^2(r) = r$ and $\alpha^{\otimes 2}(r') = r'$.  In the last equality we used the anti-symmetry of $[-,-]$ and $\alpha([-,-]) = [-,-] \circ \alpha^{\otimes 2}$.  Similar computations give
\begin{subequations}
\label{BCD}
\allowdisplaybreaks
\begin{align}
B_2 &= [\alpha(r_2),r_2'] \otimes \alpha([x,r_1]) \otimes \alpha(r_1')
= \alpha^{\otimes 3}\left\{\alpha([r_1,r_1']) \otimes [x,\alpha(r_2)] \otimes \alpha^2(r_2')\right\},\label{B2}\\
C_3 &= [[x,r_2],r_1'] \otimes \alpha(r_2') \otimes \alpha^2(r_1) = [[r_2',x],\alpha(r_2)] \otimes \alpha^2(r_1) \otimes \alpha^2(r_1'),\label{C3}\\
D_2 &= [[x,r_2],r_2'] \otimes \alpha^2(r_1) \otimes \alpha(r_1') = [[x,r_2],\alpha(r_2')] \otimes \alpha^2(r_1) \otimes \alpha^2(r_1').\label{D2}
\end{align}
\end{subequations}
Using, in addition, the anti-symmetry and the Hom-Jacobi identity of $[-,-]$, we add $C_3$ and $D_2$:
\begin{equation}
\label{C3D2}
\begin{split}
C_3 + D_2 &=  \left\{[[r_2',x],\alpha(r_2)] + [[x,r_2],\alpha(r_2')]\right\} \otimes \alpha^2(r_1) \otimes \alpha^2(r_1')\\
&= [\alpha(x),[r_2,r_2']] \otimes \alpha^2(r_1) \otimes \alpha^2(r_1')\\
&= [\alpha(x),[r_1,r_1']] \otimes \alpha^2(r_2) \otimes \alpha^2(r_2')\\
&=[\alpha(x),[\alpha(r_1),\alpha(r_1')]] \otimes \alpha^3(r_2) \otimes \alpha^3(r_2')\\
&= \alpha^{\otimes 3}\left\{[x,[r_1,r_1']] \otimes \alpha^2(r_2) \otimes \alpha^2(r_2')\right\}.
\end{split}
\end{equation}
Combining \eqref{A3}, \eqref{B2}, and \eqref{C3D2} and using the definition \eqref{eq:ad} of $\ad_x$, we now conclude that
\[
\begin{split}
A_3 + B_2 + C_3 + D_2
&= \alpha^{\otimes 3} \left\{[x,[r_1,r_1']] \otimes \alpha^2(r_2) \otimes \alpha^2(r_2')\right\} + \alpha^{\otimes 3}\left\{\alpha([r_1,r_1']) \otimes [x,\alpha(r_2)] \otimes \alpha^2(r_2')\right\}\\
&\relphantom{} + \alpha^{\otimes 3}\left\{\alpha([r_1,r_1']) \otimes \alpha^2(r_2) \otimes [x,\alpha(r_2')]\right\}\\
&= \alpha^{\otimes 3}\left\{\ad_x([r_1,r_1'] \otimes \alpha(r_2) \otimes \alpha(r_2'))\right\}\\
&= \alpha^{\otimes 3}(\ad_x([r_{12},r_{13}])),
\end{split}
\]
which proves \eqref{ad1213}.

The equalities \eqref{ad1223} and \eqref{ad1323} are proved by very similar computations.  Therefore, the equality \eqref{cojacobiad} holds.  Together with \eqref{cojacobi} we have shown that the Hom-co-Jacobi identity of $\Delta = \ad(r)$ is equivalent to $\alpha^{\otimes 3}(\ad_x([[r,r]]^\alpha)) = 0$.
\end{proof}

The following result is an immediate consequence of Theorem ~\ref{thm:cob-char}.  It gives sufficient conditions under which a Hom-Lie algebra becomes a quasi-triangular Hom-Lie bialgebra.

\begin{corollary}
\label{cor:cob}
Let $(L,[-,-],\alpha)$ be a Hom-Lie algebra and $r \in L^{\otimes 2}$ be an element such that $\alpha^{\otimes 2}(r) = r$, $r_{21} = -r$, and $[[r,r]]^\alpha = 0$ \eqref{eq:chybe}.  Then $(L,[-,-],\ad(r),\alpha,r)$ is a quasi-triangular Hom-Lie bialgebra.
\end{corollary}


To end this section, we provide several equivalent characterizations of the CHYBE \eqref{eq:chybe} in a coboundary Hom-Lie bialgebra.  Let us first define some maps that will be used in the following result.  Fix a coboundary Hom-Lie bialgebra $L = (L,[-,-],\Delta,\alpha,r)$ with $r = \sum r_1 \otimes r_2$.  Recall that $L^* = \Hom(L,\bk)$ is the linear dual of $L$.  Define the linear maps $\rho_1,\rho_2,\lambda_1,\lambda_2 \colon L^* \to L$ as follows:
\begin{equation}
\label{rholambda}
\begin{split}
\rho_1(\phi) &= \langle \phi,\alpha(r_1)\rangle r_2, \qquad
\rho_2(\phi) = \langle \phi, r_1\rangle \alpha(r_2),\\
\lambda_1(\phi) &= \alpha(r_1)\langle \phi,r_2\rangle,\qquad
\lambda_2(\phi) = r_1\langle \phi,\alpha(r_2)\rangle
\end{split}
\end{equation}
for $\phi \in L^*$.  The following result is a generalization of \cite[Lemma 8.1.6]{majid}, which deals with coboundary Lie bialgebras.

\begin{theorem}
\label{thm:quasi}
Let $(L,[-,-],\Delta,\alpha,r)$ be a coboundary Hom-Lie bialgebra.  Then the following statements are equivalent, in which the last two statements only apply when $L$ is finite dimensional.
\begin{enumerate}
\item
$L$ is a quasi-triangular Hom-Lie bialgebra, i.e., $[[r,r]]^\alpha = 0$ \eqref{eq:chybe}.
\item
$(\alpha \otimes \Delta)(r) = -[r_{12},r_{13}]$ \eqref{eq:r123}.
\item
$(\Delta \otimes \alpha)(r) = [r_{13},r_{23}]$ \eqref{eq:r123}.
\item
The diagram
\begin{equation}
\label{rhobracket}
\SelectTips{cm}{10}
\xymatrix{
L^* \otimes L^* \ar[r]^-{[-,-]} \ar[d]_{\rho_1^{\otimes 2}} & L^* \ar[d]^{\rho_2}\\
L \otimes L \ar[r]^-{[-,-]} & L
}
\end{equation}
commutes, where the bracket in $L^*$ is defined as in \eqref{eq:dualmaps}.
\item
The diagram
\begin{equation}
\label{lambdabracket}
\SelectTips{cm}{10}
\xymatrix{
L^* \otimes L^* \ar[r]^-{[-,-]} \ar[d]_{\lambda_2^{\otimes 2}} & L^* \ar[d]^{\lambda_1}\\
L \otimes L \ar[r]^-{-[-,-]} & L
}
\end{equation}
commutes.
\item
The diagram
\begin{equation}
\label{rhoDelta}
\SelectTips{cm}{10}
\xymatrix{
L^* \ar[r]^-{\Delta} \ar[d]_{\rho_1} & L^* \otimes L^* \ar[d]^{\rho_2^{\otimes 2}}\\
L \ar[r]^-{-\Delta} & L \otimes L
}
\end{equation}
commutes, where the cobracket $\Delta$ on $L^*$ is defined as in \eqref{eq:dualmaps}.
\item
The diagram
\begin{equation}
\label{lambdaDelta}
\SelectTips{cm}{10}
\xymatrix{
L^* \ar[r]^-{\Delta} \ar[d]_{\lambda_2} & L^* \otimes L^* \ar[d]^{\lambda_1^{\otimes 2}}\\
L \ar[r]^-{\Delta} & L \otimes L
}
\end{equation}
commutes.
\end{enumerate}
\end{theorem}

\begin{proof}
The equivalence between the first three statements clearly follows from
\begin{equation}
\label{alphaDelta}
(\alpha \otimes \Delta)(r) = [r_{12},r_{23}] + [r_{13},r_{23}] \quad\text{and}\quad (\Delta \otimes \alpha)(r) = - [r_{12},r_{13}] - [r_{12},r_{23}].
\end{equation}
To see that \eqref{alphaDelta} holds, let $r' = \sum r_1' \otimes r_2'$ be another copy of $r$.  Since $\Delta = \ad(r)$ \eqref{eq:Delta=ad}, the first equality in \eqref{alphaDelta} holds because:
\[
\begin{split}
(\alpha \otimes \Delta)(r_1 \otimes r_2)
&= \alpha(r_1) \otimes [r_2,r_1'] \otimes \alpha(r_2') + \alpha(r_1) \otimes \alpha(r_1') \otimes [r_2,r_2']\\
&= [r_{12},r_{23}] + [r_{13},r_{23}].
\end{split}
\]
The second equality in \eqref{alphaDelta} is proved similarly.  In view of the definitions \eqref{eq:chybe} and \eqref{eq:r123}, the equalities in \eqref{alphaDelta} imply that the first three statements in the Theorem are equivalent.

Next we show the equivalence between statements (1), (4), and (6).  Indeed, the CHYBE (i.e., $[[r,r]]^\alpha = 0$) holds if and only if
\[
\langle \phi \otimes \psi \otimes Id, -[[r,r]]^\alpha \rangle = 0
\]
for all $\phi, \psi \in L^*$.  Using the second equality in \eqref{alphaDelta} and the definition of $[-,-]$ in $L^*$ \eqref{eq:dualmaps}, we compute as follows:
\[
\allowdisplaybreaks
\begin{split}
\langle \phi \otimes \psi \otimes Id, -[[r,r]]^\alpha \rangle
&= \langle \phi \otimes \psi \otimes Id, (\Delta \otimes \alpha)(r) - \alpha(r_1) \otimes \alpha(r_1') \otimes [r_2,r_2']\rangle\\
&= \langle [\phi,\psi],r_1\rangle \alpha(r_2) - \langle\phi,\alpha(r_1)\rangle\langle\psi,\alpha(r_1')\rangle [r_2,r_2']\\
&= \rho_2([\phi,\psi]) - [\rho_1(\phi),\rho_1(\psi)].
\end{split}
\]
The last line above is equal to zero if and only if the square \eqref{rhobracket} is commutative.  This shows that statements (1) and (4) are equivalent.

Now we show the equivalence between statements (1) and (6).  The the CHYBE ($[[r,r]]^\alpha = 0$) holds if and only if
\[
\langle \phi \otimes Id \otimes Id, [[r,r]]^\alpha\rangle = 0
\]
for all $\phi \in L^*$.  Using the first equality in \eqref{alphaDelta} and the definition of $\Delta$ in $L^*$ \eqref{eq:dualmaps}, we compute as follows, where $\Delta(\phi) = \sum \phi_1 \otimes \phi_2$:
\[
\allowdisplaybreaks
\begin{split}
\langle \phi \otimes Id \otimes Id, [[r,r]]^\alpha\rangle
&= \langle \phi \otimes Id \otimes Id, [r_1,r_1'] \otimes \alpha(r_2) \otimes \alpha(r_2') + (\alpha \otimes \Delta)(r)\rangle\\
&= \langle \phi_1,r_1\rangle\langle \phi_2,r_1'\rangle \alpha(r_2) \otimes \alpha(r_2') + \langle \phi,\alpha(r_1)\rangle\Delta(r_2)\\
&= \rho_2^{\otimes 2}(\Delta(\phi)) + \Delta(\rho_1(\phi)).
\end{split}
\]
The last line above is equal to zero if and only if the square \eqref{rhoDelta} is commutative.  This shows that statements (1) and (6) are equivalent.  The equivalence between statements (1), (5), and (7) is proved similarly.
\end{proof}

\section{Cobracket perturbation in Hom-Lie bialgebras}
\label{sec:perturb}

The purpose of this final section is to study perturbation of cobrackets in Hom-Lie bialgebras, following Drinfel'd's perturbation theory of quasi-Hopf algebras \cite{dri83b,dri89b,dri90,dri91,dri92}.  The basic question is this:
\begin{quote}
If $(L,[-,-],\Delta,\alpha)$ is a Hom-Lie bialgebra (Definition ~\ref{def:hlb}) and $t \in L^{\otimes 2}$, under what conditions does the perturbed cobracket $\Delta_t = \Delta + \ad(t)$ give another Hom-Lie bialgebra $(L,[-,-],\Delta_t,\alpha)$?
\end{quote}
This is a natural question because $\Delta$ is a $1$-cocycle (Remark ~\ref{rk:cocycle}), $\ad(t)$ \eqref{eq:ad} is a $1$-coboundary provided $\alpha^{\otimes 2}(t) = t$ (Remark ~\ref{rk:delta}), and perturbation of cocycles by coboundaries is a natural concept in homological algebra.  Of course, we have more to worry about than just the cocycle condition \eqref{eq:compatibility} because $(L,\Delta_t,\alpha)$ must be a Hom-Lie coalgebra (Definition ~\ref{def:hlc}).


In the following result, we give some sufficient conditions under which the perturbed cobracket $\Delta_t$ gives another Hom-Lie bialgebra.  This is a generalization of \cite[Theorem 8.1.7]{majid}, which deals with cobracket perturbation in Lie bialgebras.  A result about cobracket perturbation in a quasi-triangular Hom-Lie bialgebra is given after the following result.  We also briefly discuss \emph{triangular} Hom-Lie bialgebra, which is the Hom version of Drinfel'd's triangular Lie bialgebra \cite{dri87}.

Let us recall some notations first.  For $t = \sum t_1 \otimes t_2 \in L^{\otimes 2}$, the symbol $t_{21}$ denotes $\tau(t) = \sum t_2 \otimes t_1$.  We extend the notation in \eqref{absolute} as follows: If $f(x,y)$ is an expression in the elements $x$ and $y$, we set
\[
|f(x,y)| = f(x,y) - f(y,x).
\]
For example, the compatibility condition ~\eqref{eq:compatibility} is $\Delta([x,y]) = |\ad_{\alpha(x)}(\Delta(y))|$, and the Hom-Jacobi identity \eqref{eq:HomJacobi} is equivalent to $[[x,y],\alpha(z)] = \left|[\alpha(x),[y,z]]\right|$.  Note that we have $|f(x,y) + g(x,y)| = |f(x,y)| + |g(x,y)|$.  Also recall the adjoint map $\ad_x \colon L^{\otimes n} \to L^{\otimes n}$ \eqref{eq:ad}.

\begin{theorem}
\label{thm:perturb}
Let $(L,[-,-],\Delta,\alpha)$ be a Hom-Lie bialgebra and $t \in L^{\otimes 2}$ be an element such that $\alpha^{\otimes 2}(t) = t$, $t_{21} = -t$, and
\begin{equation}
\label{eq:perturbedjacobi}
\alpha^{\otimes 3}\left(\ad_x([[t,t]]^\alpha + \cyclicsum (\alpha \otimes \Delta)(t))\right) = 0
\end{equation}
for all $x \in L$.  Define the perturbed cobracket $\Delta_t = \Delta + \ad(t)$.  Then $L_t = (L,[-,-],\Delta_t,\alpha)$ is a Hom-Lie bialgebra.
\end{theorem}

\begin{proof}
Since $(L,[-,-],\alpha)$ is a Hom-Lie algebra, to show that $L_t$ is a Hom-Lie bialgebra, we need to prove four things: (i) $\alpha^{\otimes 2} \circ \Delta_t = \Delta_t \circ \alpha$, (ii) $\Delta_t$ is anti-symmetric, (iii) the compatibility condition \eqref{eq:compatibility} holds for $\Delta_t$ and $[-,-]$, and (iv)  $\Delta_t$ satisfies the Hom-co-Jacobi identity (Definition ~\ref{def:hlc}).  We will reuse part of the proof of Theorem ~\ref{thm:cob-char}.

For (i), we know that $\alpha$ commutes with $\ad(t)$, which was established in the second paragraph in the proof of Theorem ~\ref{thm:cob-char}.  Since $\alpha$ commutes with $\Delta$ already, we conclude that it commutes with $\Delta_t = \Delta + \ad(t)$ as well, proving (i).

Next we consider (ii), the anti-symmetry of $\Delta_t$.  Since $\Delta$ is already anti-symmetric, $\Delta_t$ is anti-symmetric if and only if $\ad(t)$ is so.  As we already proved in \eqref{deltaantisym}, the anti-symmetry of $\ad(t)$ follows from the assumption $t + t_{21} = 0$.

Now we explain why (iii) (the compatibility condition \eqref{eq:compatibility} for $\Delta_t$ and $[-,-]$) holds.  We need to show that
\begin{equation}
\label{eq:deltatxy}
\Delta_t([x,y]) = |\ad_{\alpha(x)}(\Delta_t(y))|.
\end{equation}
Since $\Delta_t = \Delta + \ad(t)$, \eqref{eq:deltatxy} is equivalent to
\[
\begin{split}
\Delta([x,y]) + \ad_{[x,y]}(t) &= \left|\ad_{\alpha(x)}(\Delta(y)) + \ad_{\alpha(x)}(\ad_y(t))\right|\\
&= \left|\ad_{\alpha(x)}(\Delta(y))\right| + \left|\ad_{\alpha(x)}(\ad_y(t))\right|
\end{split}
\]
Since $\Delta([x,y]) = \left|\ad_{\alpha(x)}(\Delta(y))\right|$ (because $L$ is a Hom-Lie bialgebra), \eqref{eq:deltatxy} is equivalent to
\begin{equation}
\label{adxy}
\ad_{[x,y]}(t) = \left|\ad_{\alpha(x)}(\ad_y(t))\right|,
\end{equation}
which holds by Lemma ~\ref{lem:perturb}.

Finally, we consider (iv), the Hom-co-Jacobi identity of $\Delta_t$, which states
\begin{equation}
\label{hcj}
\cyclicsum (\alpha \otimes \Delta_t)(\Delta_t(x)) = 0
\end{equation}
for all $x \in L$.  Using the definition $\Delta_t = \Delta + \ad(t)$, we can rewrite \eqref{hcj} as
\begin{equation}
\label{hcj2}
0 = \cyclicsum \left\{(\alpha \otimes \Delta)(\Delta(x)) + (\alpha \otimes \ad(t))(\Delta(x)) + (\alpha \otimes \Delta)(\ad_x(t)) + (\alpha \otimes \ad(t))(\ad_x(t))\right\}.
\end{equation}
We already know that $\cyclicsum (\alpha \otimes \Delta)(\Delta(x)) = 0$, which is the Hom-co-Jacobi identity of $\Delta$.  Moreover, in \eqref{cojacobi} and \eqref{cojacobiad} (in the proof of Theorem ~\ref{thm:cob-char} with $t$ instead of $r$), we already showed that
\begin{equation}
\label{hcj3}
\cyclicsum (\alpha \otimes \ad(t))(\ad_x(t)) = \alpha^{\otimes 3}(\ad_x([[t,t]]^\alpha)).
\end{equation}
In view of \eqref{hcj2} and \eqref{hcj3}, the Hom-co-Jacobi identity of $\Delta_t$ \eqref{hcj} is equivalent to
\begin{equation}
\label{hcj4}
0 = \alpha^{\otimes 3}(\ad_x([[t,t]]^\alpha)) + \cyclicsum \left\{(\alpha \otimes \ad(t))(\Delta(x)) + (\alpha \otimes \Delta)(\ad_x(t))\right\}.
\end{equation}
Using the assumption ~\eqref{eq:perturbedjacobi}, the condition \eqref{hcj4} is equivalent to
\begin{equation}
\label{hcj5}
\cyclicsum \left\{(\alpha \otimes \ad(t))(\Delta(x)) + (\alpha \otimes \Delta)(\ad_x(t))\right\}
= \alpha^{\otimes 3}\left(\cyclicsum \ad_x((\alpha \otimes \Delta)(t))\right).
\end{equation}
We will prove \eqref{hcj5} in Lemma ~\ref{lem2:perturb} below.

The proof of Theorem \ref{thm:perturb} will be complete once we prove Lemma ~\ref{lem2:perturb}.
\end{proof}

\begin{lemma}
\label{lem2:perturb}
The condition \eqref{hcj5} holds.
\end{lemma}

\begin{proof}
Write $\Delta(x) = \sum x_1 \otimes x_2$.  Then the left-hand side of \eqref{hcj5} is:
\begin{subequations}
\allowdisplaybreaks
\begin{align*}
& \cyclicsum \left\{(\alpha \otimes \ad(t))(\Delta(x)) + (\alpha \otimes \Delta)(\ad_x(t))\right\}\\
&= \cyclicsum \left\{\alpha(x_1) \otimes \ad_{x_2}(t_1\otimes t_2) + (\alpha \otimes \Delta)([x,t_1] \otimes \alpha(t_2) + \alpha(t_1) \otimes [x,t_2])\right\}\\
&= \cyclicsum \left\{\alpha(x_1) \otimes [x_2,t_1] \otimes \alpha(t_2) + \alpha(x_1) \otimes \alpha(t_1) \otimes [x_2,t_2]\right\}\\
&\relphantom{} + \cyclicsum \left\{\alpha([x,t_1]) \otimes \Delta(\alpha(t_2)) + \alpha^2(t_1) \otimes \Delta([x,t_2])\right\}\\
\intertext{Write $\Delta(t_2) = \sum t_2' \otimes t_2''$.  Recall that $\Delta([x,t_2]) = \ad_{\alpha(x)}(\Delta(t_2)) - \ad_{\alpha(t_2)}(\Delta(x))$ \eqref{eq:compatibility} (because $L$ is a Hom-Lie bialgebra).  Making use of the fact that we have a cyclic sum, we can continue the above computation as follows:}
&= \cyclicsum \left\{\alpha(x_1) \otimes [x_2,t_1] \otimes \alpha(t_2) + \alpha(x_1) \otimes \alpha(t_1) \otimes [x_2,t_2] + \alpha([x,t_1]) \otimes \alpha^{\otimes 2}(\Delta(t_2))\right\}\\
&\relphantom{} + \cyclicsum \left\{\alpha^2(t_1) \otimes [\alpha(x),t_2'] \otimes \alpha(t_2'') + \alpha^2(t_1) \otimes \alpha(t_2') \otimes [\alpha(x),t_2'']\right\}\\
&\relphantom{} - \cyclicsum \left\{\alpha(x_2) \otimes \alpha^2(t_1) \otimes [\alpha(t_2),x_1] + \alpha(x_1) \otimes [\alpha(t_2),x_2] \otimes \alpha^2(t_1)\right\}\\
\intertext{It follows from the anti-symmetry of $\Delta$ applied to $x$ (i.e., $\sum x_2 \otimes x_1 = - \sum x_1 \otimes x_2$), $t_{21} = -t$, and $\alpha^{\otimes 2}(t) = t$ that the first two terms and the last two terms above cancel out.  Using the commutation of $\alpha$ with $[-,-]$ and $\Delta$ and $\alpha^{\otimes 2}(t) = t$, the above computation continues as follows:}
&= \cyclicsum \left\{\alpha([x,t_1]) \otimes \alpha^{\otimes 2}(\Delta(t_2)) + \alpha^2(t_1) \otimes [\alpha(x),t_2'] \otimes \alpha(t_2'') + \alpha^2(t_1) \otimes \alpha(t_2') \otimes [\alpha(x),t_2'']\right\}\\
&= \cyclicsum \left\{\alpha([x,\alpha(t_1)]) \otimes \alpha^{\otimes 2}(\Delta(\alpha(t_2))) + \alpha^3(t_1) \otimes [\alpha(x),\alpha(t_2')] \otimes \alpha^2(t_2'')\right\}\\
&\relphantom{} + \cyclicsum \left\{\alpha^3(t_1) \otimes \alpha^2(t_2') \otimes [\alpha(x),\alpha(t_2'')]\right\}\\
&= \alpha^{\otimes 3}\left(\cyclicsum \left\{[x,\alpha(t_1)] \otimes \Delta(\alpha(t_2)) + \alpha^2(t_1) \otimes [x,t_2'] \otimes \alpha(t_2'') + \alpha^2(t_1) \otimes \alpha(t_2') \otimes [x,t_2'']\right\}\right)\\
&= \alpha^{\otimes 3}\left(\cyclicsum \ad_x\left(\alpha(t_1) \otimes t_2' \otimes t_2''\right)\right)\\
&= \alpha^{\otimes 3}\left(\cyclicsum \ad_x((\alpha\otimes\Delta)(t))\right).
\end{align*}
\end{subequations}
This proves \eqref{hcj5}.
\end{proof}

The following result is a special case of the previous Theorem.

\begin{corollary}
\label{cor1:perturb}
Let $(L,[-,-],\Delta,\alpha)$ be a Hom-Lie bialgebra and $t \in L^{\otimes 2}$ be an element such that $\alpha^{\otimes 2}(t) = t$, $t_{21} = -t$, and $[[t,t]]^\alpha + \cyclicsum (\alpha \otimes \Delta)(t) = 0$.  Then $L_t = (L,[-,-],\Delta_t = \Delta + \ad(t),\alpha)$ is a Hom-Lie bialgebra.
\end{corollary}

The following result gives sufficient conditions under which the cobracket in a quasi-triangular Hom-Lie bialgebra (Definition ~\ref{def:coboundary}) can be perturbed to give another quasi-triangular Hom-Lie bialgebra.

\begin{corollary}
\label{cor:perturb}
Let $(L,[-,-],\Delta=\ad(r),\alpha,r)$ be a quasi-triangular Hom-Lie bialgebra and $t \in L^{\otimes 2}$ be an element such that $\alpha^{\otimes 2}(t) = t$, $t_{21} = -t$,
\[
[[t,t]]^\alpha + \cyclicsum (\alpha \otimes \Delta)(t)) = 0, \quad\text{and}\quad [[r,t]]^\alpha + [[t,r]]^\alpha + [[t,t]]^\alpha = 0.
\]
Then $L_t = (L,[-,-],\Delta_t = \ad(r+t),\alpha,r+t)$ is a quasi-triangular Hom-Lie bialgebra.
\end{corollary}

\begin{proof}
Indeed, Corollary ~\ref{cor1:perturb} implies that $L_t$ is a coboundary Hom-Lie bialgebra (Definition ~\ref{def:coboundary}), since $\ad(r) + \ad(t) = \ad(r+t)$ and $\alpha^{\otimes 2}(r + t) = \alpha^{\otimes 2}(r) + \alpha^{\otimes 2}(t) = r+t$.  The sum $r+t$ satisfies the CHYBE \eqref{eq:chybe} because
\[
[[r+t,r+t]]^\alpha = [[r,r]]^\alpha + [[r,t]]^\alpha + [[t,r]]^\alpha + [[t,t]]^\alpha
\]
and $r$ satisfies the CHYBE (i.e., $[[r,r]]^\alpha = 0$).
\end{proof}

Let us give an interpretation of the previous Corollary.  Define a \textbf{triangular Hom-Lie bialgebra} as a quasi-triangular Hom-Lie bialgebra $(L,[-,-],\Delta=\ad(t),\alpha,t)$ (Definition ~\ref{def:coboundary}) in which $t$ is anti-symmetric (i.e., $t_{21} = -t$).  A triangular Hom-Lie bialgebra with $\alpha = Id$ is exactly a triangular Lie bialgebra, as defined by Drinfel'd \cite{dri87}.  Corollary ~\ref{cor:perturb} implies that every triangular Hom-Lie bialgebra is obtained as a perturbation of the trivial cobracket $\Delta = \ad(0)$.  Of course, one can infer this fact from Corollary ~\ref{cor:cob} as well.


\end{document}